\documentclass{article}
\usepackage{amsfonts}
\usepackage{amssymb}
\usepackage[dvips]{color}
\usepackage{graphicx}
\usepackage{amsmath}

\newtheorem{theorem}{Theorem}

\newtheorem{properties}[theorem]{Properties}

\newtheorem{definition}[theorem]{Definition}
\newtheorem{example}[theorem]{Example}

\newtheorem{proposition}[theorem]{Proposition}
\newtheorem{remark}[theorem]{Remark}

\newenvironment{proof}[1][Proof]{\textbf{#1.} }{\ \rule{0.5em}{0.5em}}

\begin{document}

\title{Disentangling Price, Risk and Model Risk: \\
{V\&R Measures}}
\author{Marco Frittelli \qquad \\
{\small Milano University, email: marco.frittelli@unimi.it} \and Marco
Maggis\qquad \\
{\small Milano University, email: marco.maggis@unimi.it} }
\maketitle

\begin{abstract}
We propose a method to assess the intrinsic risk carried by a financial
position $X$ when the agent faces uncertainty about the pricing rule
assigning its present value. Our approach is inspired by a new
interpretation of the quasiconvex duality in a Knightian setting, where a
family of probability measures replaces the single reference probability and
is then applied to value financial positions.

Diametrically, our construction of Value\&Risk measures is based on the
selection of a basket of claims to test the reliability of models. We
compare a random payoff $X$ with a given class of derivatives written on $X$%
, and use these derivatives to \textquotedblleft test\textquotedblright\ the
pricing measures.

We further introduce and study a general class of Value\&Risk measures $%
R(p,X,\mathbb{P})$ that describes the additional capital that is required to
make $X$ acceptable under a probability $\mathbb{P}$ and given the initial
price $p$ paid to acquire $X$.
\end{abstract}

\bigskip



\section{Introduction}

\bigskip The art of finance is essentially related to the capacity of
transferring the Risk: many notions (replicability, hedging trading
strategies, superhedging, quantile hedging, partial hedging, indifference
pricing, see for example \cite{FoSch}) are essentially based on some
technique which aims at \textit{replacing} the risk carried by one claim $X$
by the risk of some other object $Y$ that is considered \textit{sufficiently
close} to $X$ (whatever it means), provided that the risk of the auxiliary
object $Y$ is easier to compute.

\medskip

In this paper we take such approach in order to evaluate the intrinsic risk
of a claim $X$ by comparing the value of $X$ with the value of a family of
derivatives $f(X)$ on $X$, having a bounded level of risk. In this way, we
will conclude that the intrinsic risk of $X$ corresponds to the maximal risk
reduction we would obtain buying $X$ at price $p$ and selling a derivative $%
f(X)$, in the given class, with a price at most equal to $p$. This
methodology is sketched below but will be analyzed in detail only in Section %
\ref{secVR}, as in the Introduction we will illustrate the main concepts
only and defer to the subsequent sections the precise notations and
mathematical details.

\bigskip

In the literature the approaches used are mainly based on the selection of a
set of \textquotedblleft calibrated\textquotedblright\ pricing model. In
this setting, an important contribution is provided by Cont \cite{Co06},
where a quantitative framework to assess Model Uncertainty was introduced.
The prices of a set of benchmark instruments written on the underlying was
supposed to be known (allowing the possibility to belong to the bid-ask
interval). Consequently arbitrage-free pricing models $\mathcal{Q}$
consistent with these benchmark prices lead to the natural definition of
Coherent Measure of Model Risk as: $\mu _{\mathcal{Q}}(\cdot )=\sup_{Q\in
\mathcal{Q}}E_{Q}[\cdot ]-\inf_{Q\in \mathcal{Q}}E_{Q}[\cdot ]$. \newline
The absolute and relative measures of model risk, based on the specification
of a set of alternative distributions around a reference one and on a worst-
and best-case approach, are introduced in Barrieu and Scandolo \cite{BS15}.
\newline
Both the approaches in \cite{Co06} and \cite{BS15} are however very
different from our analysis developed in Section \ref{secVR}.

\bigskip

We let $\mathcal{L}(\Omega ,\mathcal{F})$ be the space of $\mathcal{F}$
measurable finite valued random variables with $\mathcal{F\subseteq B}%
(\Omega )$, the Borel sigma algebra of a Polish space $\Omega $. If $f:%
\mathbb{R}\rightarrow \mathbb{R}$ is a Borel function and $X\in \mathcal{L}%
(\Omega ,\mathcal{F})$, the random variable $f(X)$ is interpreted as the
terminal payoff of a contingent claim written on the underlying asset having
terminal value $X$. Suppose that the price of this contingent claim is
determined by the real function $f$ and by the distribution function $%
\mathbb{Q}(X\leq x)$ of $X$ with respect to a \textquotedblleft
pricing\textquotedblright\ probability measure $\mathbb{Q}$. As the choice
of such pricing measure is clearly an important and problematic issue, in
our approach we will contemplate a \textit{model risk} function defined on a
set of plausible models. The price under $\mathbb{Q}$ will be given by the
formula:%
\begin{equation*}
E_{\mathbb{Q}}[f(X)]=\int f_{X}d\mathbb{Q}=\int fdQ_{X}\text{,}
\end{equation*}%
where $f_{X}:=f\circ X$ is the random variable $f(X)$ defined on $(\Omega ,%
\mathcal{F})$ and $Q_{X}:=\mathbb{Q}\circ X^{-1}$ is the law of $X$ under $%
\mathbb{Q}$. \newline
The reason of writing explicitly the above formula is that in the two
approaches below we will simply exploit the \textquotedblleft bilinear
form\textquotedblright\ $\left\langle f,\mathbb{Q}\right\rangle _{X}:=E_{%
\mathbb{Q}}[f(X)]$, testing one variable via a set of the dual (testing)
variables. \newline
We stress the analogy of the two approaches that will be developed in
Sections 4.1 and 4.2 and that are here briefly introduced.

\paragraph{Use Models to test Claims\label{sec1}.}

Consider an underlying $X\in \mathcal{L}(\Omega ,\mathcal{F})$ and a claim $%
f:\mathbb{R}\rightarrow \mathbb{R}$.

In this approach, we \textquotedblleft compare\textquotedblright\ the prices
of the contingent claim $f(X)$ with respect to a given class of probability
models $\mathbb{Q}\in \mathcal{M}$, and use these models to
\textquotedblleft test\textquotedblright\ $f(X)$. In other words we take the
classical Knightian Uncertainty point of view and adopt a set of probability
measures $\mathbb{Q}\in \mathcal{M}$ to asses possible prices of the claim.
This idea is in agreement with the definition proposed by Cont \cite{Co06}:
the range of feasible prices varies from the minimal to the maximal one.
Indeed an agent may incorporate her preferences, binding a maximal model
risk she is willing to accept when choosing a pricing probability.

In our approach we further assume the existence of a model-risk function $A$
on $\mathcal{M}$ so that we may define the best (seller) price of the claim $%
f(X)$ relative to all possible choices of pricing measures under
the constraint that the model risk is less than or equal to $a$,
formally $$ V(a,X;f)=\sup_{\mathbb{Q}\in \mathcal{M}}\left\{ E_{\mathbb{Q}}[f(X)]\mid A(%
\mathbb{Q},X)\leq a\right\}.$$ In this way $V(a,X;f)$ represents
the maximum value of the contingent claim $f$ on the underlying
$X$, for the level $a$ of model-risk in the choice of $Q$, i.e. is
the best (seller) price of the claim $f(X)$ relative to all
possible choices of pricing measures under the constraint that the
model risk is less than or equal to $a$. By applying results from
quasi-convex duality (see \cite{FMP12}) we then show under which
conditions it is possible to recover the model risk function $A$
from the inverse function of $V.$

\paragraph{Use Claims to test Models\label{sec2}.}

Consider an underlying $X\in \mathcal{L}(\Omega ,\mathcal{F})$ and a
probability $\mathbb{Q}\in \mathcal{P}(\Omega )$.

In this novel approach, we \textquotedblleft compare\textquotedblright\ $X$
with the derivatives on $X$, in a given class of derivatives $f\in K$, and
use these derivatives to \textquotedblleft test\textquotedblright\ $\mathbb{Q%
}$. Contrary to the above mentioned (Knightian Uncertainty)
approach, here we select a class $K$ of derivatives to test the
\textquotedblleft reliability\textquotedblright\ of the model\
$\mathbb{Q}$. \newline An agent is willing to hold (or sell) the
position $X$ but she is aware that she may face losses. In order
to control these potential losses she will try to transfer/reduce
the overall risk by buying (or selling) derivatives/insurances on
$X$. We assume the existence of a risk reduction function,
$\varphi (f,X)$, defined on the basket of claims $K$ (which in
general is independent from a particular choice of the reference
probability). Among those derivatives, which guarantee the same
level of risk reduction $r$, the agent will choose the cheapest
one with respect to the pricing rule she adopt, computing the
minimal price

\begin{equation*}
\Pi _{\varphi }(r,X,;\mathbb{Q}):=\inf_{f\in K}\left\{ E_{\mathbb{Q}}\left[
f(X)\right] \mid \varphi (f,X)\geq r\right\} .
\end{equation*}%
The intrinsic risk for $X$, given that its present value is $p$ and $\mathbb{%
Q}$ is the pricing rule selected by the agent, is therefore provided by the
left inverse of $\Pi ,$ namely

\begin{equation}
R_{\varphi }(p,X;\mathbb{Q}):=\sup \{s\in \mathbb{R}\mid \Pi _{\varphi }(s,X;%
\mathbb{Q})\leq p\}.  \label{RR}
\end{equation}%
In fact if the optimization problems just mentioned can be solved then there
will exists a derivative $f(X)$ such that the price $E_{\mathbb{Q}}[f(X)]$
is equal to $p$ and provides a risk reduction $R_{\varphi }(p,X;\mathbb{Q})$.

We analyze several properties of the map $R_{\varphi}$ (see
Proposition \ref{properties} and \ref{cash}) including the
dependence of $R_{\varphi }$ from the set $K$ (Proposition
\ref{dependence:K}).

In Section \ref{examples} we show how the choice of the class of test
functions $K$ for $R_{\varphi }$ defined in \eqref{RR}, can be adapted to
several different contexts. The key idea is that $K$ collects those
derivatives which can be sold or acquired in order to cover
unexpected/unbounded losses of the underlying $X$. In addition, we prove in
Proposition \ref{dualityK} a quasi-convex duality result that allows us to
recover the risk reduction function $\varphi $ from $R_{\varphi }$.

\bigskip

%

To the best of our knowledge, the approach of using a fixed basket of claims
to test the reliability of models was not yet developed in the mathematical
finance literature and it represents the first main contribution of this
paper (see Section \ref{secVR}). The second one is the analysis and
axiomatization on the Value and Risk (V\&R) measures that we now illustrate.

\subsection{On Value and Risk Measures}

In Section \ref{ProfitLoss} we propose a systematic study which allows to
answer to the controversy about whether one should consider the future
\textit{value} of a position or the \textit{change in values} as the
argument of a risk measure (see the following excerpt from \cite{ADEHb}).

\bigskip

\textquotedblleft \emph{Although several papers (including an earlier
version of this one) define risk in terms of changes in values between two
dates, we argue that because risk is related to the variability of the
future value of a position, due to market changes or more generally to
uncertain events, it is better to instead consider future values only.
Notice indeed that there is no need for the initial costs of the components
of the position to be determined from universally defined market prices
(think of over-the-counter transactions). The principle of bygones are
bygones leads to this future wealth approach.}\textquotedblright , Section
2.1, Artzner et al. \cite{ADEHb}.

\bigskip

Differently from what is suggested in \cite{ADEHb}, it is a common practice
to apply standard risk measures as the Value at Risk or the Expected
Shortfall to Profit and Loss (P\&L) distributions. Given the triple $(p,X,%
\mathbb{P})$ with $p$ being the observed present value of $X$ and
$\mathbb{P} $ a reference probability, the P\&L distribution is
the induced distribution of the variation $X-p$ with respect to
$\mathbb{P}$. Indeed the P\&L approach has the benefit to
incorporate the price component in the risk assessment. On the
other hand it is not possible to distinguish which source
contributes mostly to the risk exposure, either a potential
mis-pricing of $X $ or the future realization of $X$. This is
clarified in Example \ref{inversione} where we consider two random
payoffs $X$ and $Y$ whose initial values are respectively $x,y\in
\mathbb{R}$ and show that, even if the payoff $X$ is
\textquotedblleft riskier\textquotedblright\ than $Y$ by any Risk
Measure (which is monotone decreasing with respect to the first
stochastic order), when considering the P\&L distributions of
$X-x$ and $Y-y$ the risk order may be reverted, if the price $y$
is too large.

\bigskip

To overcome this drawback in Section \ref{ProfitLoss} we will thus consider
the triple\textit{\ }%
\begin{equation*}
(p,X,\mathbb{P})\in \mathbb{R}\times \mathcal{L}(\Omega ,\mathcal{F})\times
\mathcal{P}(\Omega )
\end{equation*}%
as the argument of a \textbf{Value\&Risk Measure} $R(p,X,\mathbb{P})$, where
$p$ is the observed initial value of $X$ or is assigned by a pricing
functional. \newline
Informally, $R(p,X,\mathbb{P})$ should describe the additional capital that
is required to make $X$ acceptable under $\mathbb{P\in }\mathcal{P}(\Omega )$
and given the initial price $p$ paid to acquire $X$. We propose an axiomatic
approach to define such Value\&Risk (V\&R) measures by describing some
desirable minimal properties that $R(p,X,\mathbb{P})$ should satisfy.
\newline
Indeed risk measures defined on Profit and Loss distributions can be
recovered as a particular case in the family of Value\&Risk measures by
defining $R(p,X;\mathbb{P})=\rho _{\mathbb{P}}(X-p)$ with $\rho _{\mathbb{P}}
$ being a risk measure defined on some vector space $\mathcal{X}\subseteq
L^{0}(\Omega ,\mathcal{F},\mathbb{P})$. This case also suggests which are
the reasonable properties that a V\&R measure should satisfy. \newline
On the other hand the map $R_{\varphi }$ as defined in \eqref{RR} is a
Value\&Risk measure (see Theorem \ref{Price&Risk}), which exceeds the common
use of Profit and Loss distributions.

\bigskip

We point out that $R(p,X,\mathbb{P})$ can be interpreted:

\begin{itemize}
\item As an index of feasibility of the measure $\mathbb{P}$, with $X$
acting the role of a fixed parameter; in this case $R(p,X,\mathbb{\cdot })$
should behave as a model risk measure over the laws $P_{X}\in \mathcal{P}(%
\mathbb{R})$ (see Section \ref{secRM} for a review of such notion);

\item As a measure of the risk we are facing buying $X$ at price $p$, with $%
\mathbb{P}$ acting as the agent model belief; in this case $R(p,\cdot ,%
\mathbb{P})$ should behave as a risk measure on random variables $X\in
\mathcal{L}(\Omega ,\mathcal{F})$.
\end{itemize}

One relevant feature of considering such V\&R Measures $R(p,X,\mathbb{P})$
is the possibility to disentangle the three most important sources of
uncertainty: the price $p$ (which in general might not be unique, but rather
belong to a bid-ask interval), the random payoff $X$ and the probability $%
\mathbb{P}$. This differs to the common practice of concentrating these
three information in a unique object which is the Profit and Loss
distribution. \newline
As a consequence this reflects into the behavior of $R(p,X,\mathbb{P})$ with
respect to the addition of a cash amount $\alpha \in \mathbb{R}$. Note that
there are several reasonable properties regarding \textquotedblleft cash
invariance\textquotedblright , corresponding to the different ways one may
add cash: $R(p+\alpha ,X,\mathbb{P});$ $R(p,X+\alpha ,\mathbb{P});$ $R(p,X,%
\mathbb{P}^{\alpha })$. In Section \ref{ProfitLoss} we will explicitly
characterize the V\&R Measures satisfying three distinct cash invariance
properties and show the relevance of taking into account the initial amount $%
p$ needed to buy $X$. We will therefore conclude (Proposition \ref%
{acceptance:cash} and Remark \ref{delbaen}), that the choice between
\textquotedblleft future value only\textquotedblright\ versus
\textquotedblleft P\&L\textquotedblright\ is not arbitrary and it rests on
the type of cash invariance one is willing to accept.

\section{Risk Measures on $\mathcal{P}(\mathbb{R})$.\label{secRM}}

\paragraph{Notations.}

Let $(\Omega ,\mathcal{B}(\Omega ))$ be a probability space with $\Omega $
Polish and $\mathcal{B}(\Omega )$ the Borel sigma algebra induced by the
metric. Let $\mathcal{L}(\Omega ,\mathcal{F})$ be the space of $\mathcal{F}$
measurable finite valued random variables, endowed with the pointwise
partial order $\leq $ and $\mathcal{L}^{\infty }(\Omega ,\mathcal{F})$ its
subspace of bounded random variables. We denote respectively by $\mathcal{P}%
(\Omega ),\;\mathcal{P}(\mathbb{R})$ the set of all probability measures on $%
(\Omega ,\mathcal{B}(\Omega ))$, $(\mathbb{R},\mathcal{B}(\mathbb{R}))$.
\newline
Notice that for $\mathbb{P}\in \mathcal{P}(\Omega )$ and $X\in \mathcal{L}%
(\Omega ,\mathcal{F})$ the expectation $E_{\mathbb{P}}[X]$ might not be even
defined and for this reason we will make use of the convention $E_{\mathbb{P}%
}[X]=E_{\mathbb{P}}[X^{+}]-E_{\mathbb{P}}[X^{-}]$ with $\infty -\infty
=-\infty $.

\noindent For any Borel function $f:\mathbb{R}\rightarrow \mathbb{R}$ and $%
X\in \mathcal{L}(\Omega ,\mathcal{F})$, the random variable $f(X)$ is
interpreted as the terminal payoff of a contingent claim written on the
underlying asset having terminal value $X$. \newline
If a probability $\mathbb{P}\in \mathcal{P}(\Omega )$ is fixed we define $%
L^{0}(\Omega ,\mathcal{F},\mathbb{P})$ be the space of $\mathcal{F}$
measurable random variables that are $\mathbb{P}$ almost surely finite,
endowed with the $\mathbb{P}$-almost sure partial order $\leq _{\mathbb{P}}$%
. \newline
For any fixed $\mathbb{P}\in \mathcal{P}(\Omega )$ the random variable $X\in
\mathcal{L}(\Omega ,\mathcal{F})$ induces a probability measure $P_{X}\in
\mathcal{P}(\mathbb{R})$ by $P_{X}=\mathbb{P}\circ X^{-1}$. We refer to \cite%
{Ali} Chapter 15 for a detailed study of the convex sets $\mathcal{P}(\Omega
)$ (resp. $\mathcal{P}(\mathbb{R})$). If $\mathbb{P}(X=x)=1$ for some $x\in
\mathbb{R}$ then $P_{X}$ is the Dirac distribution, denoted by $\delta _{x},$
that concentrates the mass in the point $x\in \mathbb{R}$. Similarly we
denote by $\delta _{\omega }\in \mathcal{P}(\Omega )$ the Dirac distribution
on $\omega \in \Omega $.

\begin{definition}
\label{defOrder} We consider the following partial order for probability
measures.

\begin{enumerate}
\item[(i)] The first order stochastic dominance on $\mathcal{P}(\mathbb{R})$
is given by:
\begin{equation*}
Q\preccurlyeq _{1}P\text{ if and only if }F_{P}(x)\leq F_{Q}(x)\text{ for
all }x\in \mathbb{R},
\end{equation*}%
where $F_{P}(x)=P(-\infty ,x]$ and $F_{Q}(x)=Q(-\infty ,x]$ are the
distribution functions of $P,Q\in \mathcal{P}(\mathbb{R})$.

\item[(ii)] For any fixed $X\in \mathcal{L}(\Omega ,\mathcal{F})$ we define
the following partial order on $\mathcal{P}(\Omega )$%
\begin{equation*}
\mathbb{P}^{1}\preccurlyeq _{X}\mathbb{P}^{2}\text{ if and only if }%
P_{X}^{1}\preccurlyeq _{1}P_{X}^{2}.
\end{equation*}
\end{enumerate}
\end{definition}

Notice that when $\mathbb{P}^{1}\preccurlyeq _{X}\mathbb{P}^{2}$ then $%
\mathbb{P}^{2}$ is a safer scenario than $\mathbb{P}^{1}$ for $X$. Observe
also that, for $\mathbb{P}\in \mathcal{P}(\Omega )$ and any $X,Y\in \mathcal{%
L}(\Omega ,\mathcal{F})$, $X\leq Y$ implies $X\leq _{\mathbb{P}}Y$ which
implies $P_{X}\preccurlyeq _{1}P_{Y}$. \newline
We shall always refer to $C^{0}(\mathbb{R})=\{f:\mathbb{R}\rightarrow
\mathbb{R}\mid f\text{ continuous }\}$ and $C_{+}^{0}(\mathbb{R})=\{f\in
C^{0}(\mathbb{R})\mid f\text{ increasing }\}$. Let $C_{b}(\Omega )$ be the
space of bounded continuous function $f:\Omega \rightarrow \mathbb{R}$ and $%
ca(\Omega )$ the space of countably additive signed measures $\mu :\mathcal{B%
}(\Omega )\rightarrow \mathbb{R}$. We endow $ca(\Omega )$ with the weak$%
^{\ast }$ topology $w^{\ast }=\sigma (ca(\Omega ),C_{b}(\Omega ))$. The dual
pairing $\langle \cdot ,\cdot \rangle :C_{b}(\Omega )\times ca(\Omega
)\rightarrow \mathbb{R}$ is given by $\langle f,\mu \rangle =\int fd\mu $
and the function $\mu \mapsto \int fd\mu $ ($f\in C_{b}(\Omega )$) is $%
w^{\ast }$ continuous.

\paragraph{Risk Measures on $\mathcal{P}(\mathbb{R})$ for a fixed reference
probability.}

\label{FMP}

We refer to \cite{FMP12} for a detailed analysis of risk measures defined on
$\mathcal{P}(\mathbb{R}).$ Recall that, when $\mathbb{P}\in \mathcal{P}%
(\Omega )$ is fixed, a map $\rho _{\mathbb{P}}:L\rightarrow \overline{%
\mathbb{R}}:=\mathbb{R}\cup \left\{ -\infty \right\} \cup \left\{ \infty
\right\} $, defined on given subset $L\subseteq L^{0}(\Omega ,\mathcal{F},%
\mathbb{P}),$ is called law invariant if $X,Y\in L$ and $P_{X}=P_{Y}$
implies $\rho _{\mathbb{P}}(X)=\rho _{\mathbb{P}}(Y)$.

\noindent Therefore, when considering law invariant risk measures $\rho _{%
\mathbb{P}}:L^{0}(\Omega ,\mathcal{F},\mathbb{P})\rightarrow \overline{%
\mathbb{R}}$ it is natural to shift the problem to the set $\mathcal{P}(%
\mathbb{R})$ by defining the new map $\Phi :\mathcal{P}(\mathbb{R}%
)\rightarrow \overline{\mathbb{R}}$ \ as $\Phi (\mathbb{P}\circ X^{-1})=\rho
_{\mathbb{P}}(X)$. This map $\Phi $ is well defined on the entire $\mathcal{P%
}(\mathbb{R})$, since there exists a bi-injective relation between $\mathcal{%
P}(\mathbb{R})$ and the quotient space $\frac{L^{0}}{\sim }$ (provided that $%
(\Omega ,\mathcal{F},\mathbb{P})$ supports a random variable with uniform
distribution), where the equivalence is given by $X\sim _{\mathcal{D}}Y$ $%
\Leftrightarrow P_{X}=P_{Y}$. However, $\mathcal{P}(\mathbb{R})$ is only a
convex set and the usual operations on $\mathcal{P}(\mathbb{R})$ are not
induced by those on $L^{0}$, namely $(P_{X}+P_{Y})(A)=P_{X}(A)+P_{Y}(A)\neq
P_{X+Y}(A)$, $A\in \mathcal{B}_{\mathbb{R}}$. From \cite{FMP12} we recall
the following

\begin{definition}
\label{defRM}A Risk Measure on $\mathcal{P}(\mathbb{R})$ is a map $\Phi :%
\mathcal{P}(\mathbb{R})\rightarrow \mathbb{R}\cup \{+\infty \}$ such that:

\begin{description}
\item[(Mon)] $\Phi $ is $\preccurlyeq _{1}$-monotone decreasing: $%
P\preccurlyeq _{1}Q$ implies $\Phi (P)\geq \Phi (Q)$;

\item[(QCo)] $\Phi $ is quasi-convex: $\Phi (\lambda P+(1-\lambda )Q)\leq
\Phi (P)\vee \Phi (Q)$, $\lambda \in \lbrack 0,1].$
\end{description}
\end{definition}

Quasiconvexity can be equivalently reformulated in terms of sublevel sets: a
map $\Phi $ is quasi-convex if for every $c\in \mathbb{R}$ the set $\mathcal{%
A}_{c}=\{P\in \mathcal{P}(\mathbb{R})\mid \Phi (P)\leq c\}$ is convex.

\bigskip

As suggested by \cite{Weber}, we define the translation operator $T_{p}$ on
the set $\mathcal{P}(\mathbb{R})$ by: $T_{p}P(-\infty ,x]:=P(-\infty ,x-p]$,
for every $p\in \mathbb{R}$. Equivalently, if $P_{X}$ is the probability
distribution of a random variable $X$ we define the translation operator as $%
T_{p}P_{X}=P_{X+p}$, $p\in \mathbb{R}$. As a consequence we map the
distribution $F_{X}(x)$ into $F_{X}(x-p)$. Notice that $P\preccurlyeq
_{1}T_{p}P$ for any $p>0$. We will interpret $T_{-p}P_{X}=\mathbb{P}\circ
(X-p)^{-1}$ as the Profit and Loss distribution of the random payoff $X$
whose initial value is $p$.

\begin{definition}
We consider the following additional property for a risk measure $\Phi :%
\mathcal{P}(\mathbb{R})\rightarrow \mathbb{R}\cup \{+\infty \}$:

\begin{description}
\item[(TrI)] $\Phi $ is translation invariant if $\Phi (T_{p}P)=\Phi (P)-p$
for any $p\in \mathbb{R}.$
\end{description}
\end{definition}

Notice that (TrI) corresponds exactly to the notion of cash additivity for
risk measures defined on a space of random variables as introduced in \cite%
{ADEHb}.

\section{Value and Risk measures: $V\&R$}

\label{ProfitLoss}

We consider a simple setting, in which the risk of a financial portfolio is
evaluated over its (empirical) profit and loss (P\&L) distribution in a
one-period investment horizon. (i.e. we restrict the problem to two dates $%
t_{0}$ and $t_{1}$). For simplicity we can think of $t_{0}=0$ and $t_{1}=1$,
but a (sufficiently long) market history is supposed to be known before time
$t_{0}=0$. The risk manager can observe the present market values of a
basket of financial tradable assets at $t_{0}$ and hence she will be able to
compute the\ time $t_{0}$ price of any portfolio strategy. Tradable assets
are described by a $d$-dimensional vector of initial prices $S^{0}\in
\mathbb{R}^{d}$ and a $d$-dimensional random vector of payoffs $S:(\Omega ,%
\mathcal{F})\rightarrow \mathbb{R}^{d}$. We are implicitly assuming the
interest rate is zero or that the asset prices are already discounted. Given
a random variable $X:(\Omega ,\mathcal{F})\rightarrow \mathbb{R}$ any choice
of the (historical) probability $\mathbb{P}\in \mathcal{P}(\Omega )$ and any
price $p$ of $X$ will determine the Profit and Loss (P\&L) distribution of $%
(X-p)$, namely
\begin{equation*}
\mathbb{P}\circ (X-p)^{-1}=T_{-p}P_{X}.
\end{equation*}%
The price $p$ could represent the observed initial price of $X$ or could be
assigned via a pricing functional. In either cases, the P\&L distribution
will be given by: $T_{-p}P_{X}$. In addition, if a risk measure $\rho _{%
\mathbb{P}}$ is also assigned, then it will induce a risk measure on P\&L
distribution $\Phi :\mathcal{P}(\mathbb{R})\rightarrow \mathbb{R}\cup
\left\{ \infty \right\} $ by: $\rho _{\mathbb{P}}(X-p)=\Phi \left(
T_{-p}P_{X}\right) $.

The drawback of such P\&L approach is that usually the price component
cannot be distinguished from the distribution component and this becomes a
critical point as far as we are facing Uncertainty on the reference
probability $\mathbb{P}$.

\bigskip

\textit{We will thus consider the triple }$(p,X,\mathbb{P})\in \mathbb{R}%
\times \mathcal{L}(\Omega ,\mathcal{F})\times \mathcal{P}(\Omega )$\textit{\
as the argument of our Value}\&\textit{Risk functional }$R(p,X,\mathbb{P})$,
where the initial value of $X$ is assigned by $p$. \newline

\begin{remark}
\label{RemCall}Let $\mathbb{P}\in \mathcal{P}(\Omega )$. To better clarify
the role of the sign of the variable $p$ we consider the following simple
example: if $C=(S_{T}-k)^{+}$ is the payoff of a Call Option written on an
underlying asset $S$, then the initial value of $C$ is positive and given by
$c$. In the case we buy $C$ we will consider the triple $(c,C,\mathbb{P})$
as the argument of $R$. On the other hand if we are selling $C$ we will
consider $(-c,-C,\mathbb{P})$. Thus in general for $R(p,X,\mathbb{P})$ the
variable $p$ represents the value of $X$ at time $0$. In particular a
positive value $p>0$ represents the price we paid to hold $X$ and a negative
values $p<0$ corresponds to the amount $|p|$ we received selling $X$.
\end{remark}

\paragraph{Illustrative observations.}

Risk measurement is in general not only a binary answer to the question `is
a portfolio acceptable?'. Any risk procedure allows us to quantify the level
of risk exposure so that an extra capital requirement can be assessed to
cover future unexpected losses. In order to develop the intuition leading to
the following definition and properties of the V\&R measures we present a
common simple situation.

\bigskip

Consider a price/portfolio couple $(p,X)$ given by the selling of a call
option by an agent whose personal belief is $\mathbb{P}$ (i.e $(p,X)=(-c,-C)$
with $c>0$). Obviously we expect that any rational agent will willingly sell
$C$ if the statistical information guarantee that the risk $\Phi (\mathbb{P}%
_{-C})$ is low enough to be recovered by the amount $c$, i.e. if $R(-c,-C,%
\mathbb{P}):=-c+\Phi (\mathbb{P}_{-C})$ is non positive. Similarly an agent
who is paying $c$ to acquire $C$ will be happy to be informed that no
additional capital is required i.e. $R(c,C,\mathbb{P})=c+\Phi (\mathbb{P}%
_{C})$ is non-positive. \newline
Informally we claim that the quantity $R(p,X,\mathbb{P})$ gives the eventual
extra capital requirement the agent has to save if the level of risk $\Phi
(P_{X})$ is too high. \newline
Thus this extra capital requirement can be written in terms of acceptance
set as follows:%
\begin{eqnarray*}
R(p,X,\mathbb{P}) &=&p+\inf \{m\mid X+m\text{ is acceptable}\} \\
&=&p+\inf \{m\mid T_{m}P_{X}\in \mathcal{A}^{0}\}
\end{eqnarray*}%
where $\mathcal{A}^{0}\subseteq \mathcal{P}(\mathbb{R})$ is the set of $P\&L$
distributions that are acceptable for the regulator. But this is only a
particular case of a more general formulation that allow to conceive several
reasonable cash additivity properties for $R(p,X,\mathbb{P}).$

\bigskip

The aforementioned situation can be summarized by a decomposition of the
type
\begin{equation}
\begin{array}{ccccc}
\text{P\&L Risk} & = & \text{Price of X} & + & \text{Risk of the Payoffs} \\
R(p,X,\mathbb{P}) & = & p & + & \Phi (P_{X})%
\end{array}
\label{12}
\end{equation}%
Recall that usually regulators/risk managers focus their attention only on
the component which estimates the risk of the Profit and Loss distribution.
The interpretation is the following: the risk of the Profit and Loss
distribution $T_{-p}P_{X}$ is strongly related to the price that was paid to
hold $X$. When $R$ is defined as in (\ref{12}) the total capital requirement
will be given by the market price $p$ that was paid to acquire $X$ plus the
risk of the payoff $\Phi (P_{X})$. Notice that usually if $p$ is positive
(resp. negative) then $\Phi (P_{X})$ is expected to be negative (resp.
positive) as suggested by the example of the Call Option described in Remark %
\ref{RemCall}. The simplest example of such $V\&R$ measure is:
\begin{equation*}
R(p,X,\mathbb{Q})=p-E_{\mathbb{Q}}[X],
\end{equation*}%
which express in fact that the intrinsic risk of acquiring $X$ at a given
price $p$ is exactly the discrepancy between $p$ and $E_{\mathbb{Q}}[X]$,
assuming that $\mathbb{Q}$ is the pricing rule. However, this case and the
decomposition in (\ref{12}) may hold only in special cases of the general
family of Value\& Risk measures.

\bigskip

In this paper we generalize the usual form $p+\inf \{m\mid T_{m}P_{X}\in
\mathcal{A}^{0}\}$. To explain this generalization we consider the following
two steps.

\bigskip

\textbf{First} we consider a situation in which acceptance of a position $X$
has an explicit dependence on its price $p$. In such a case $R(p,X,\mathbb{P}%
)=\inf \{m\mid T_{m}P_{X}\in \mathcal{A}^{p}\}$, and we recover the
classical framework if $\mathcal{A}^{p}=\{T_{-p}P\mid P\in \mathcal{A}^{0}\}$%
.

\bigskip

\textbf{Second} we push the problem to the utmost general and interesting
situation where $R(p,X,\mathbb{P})=\inf \{m\mid P_{X}\in \mathcal{A}%
_{m}^{p}\}$. We would like to stress that in the definition $R_{\mathbb{A}%
}(p,X,\mathbb{P})=\inf \{m\mid T_{m}P_{X}\in \mathcal{A}^{p}\}$ the position
$T_{m}P_{X}$ has a different initial value with respect to $P_{X}$ which is
naively speaking $p+m$. Notice that the set $\mathcal{A}_{m}^{p}$ is
explicitly splitting the two components $p$ and $m$ corresponding
respectively to the \emph{initial value of the position} and the \emph{%
capital requirement to cover expected losses}. Potentially these two
components might be expressed in two different currencies and for this
reason the quantity $T_{m}P$ might loose its meaning.

\begin{example}
\label{inversione} We now consider two portfolios $X,Y$ whose initial values
are respectively $x,y$ and suppose that the distribution of $X$ dominates
the one of $Y,$ $P_{X}\preccurlyeq _{1}P_{Y}$ (which informally means $X$ is
\textquotedblleft riskier\textquotedblright\ than $Y$) and therefore for any
Risk Measure $\Phi :\mathcal{P}\rightarrow \overline{\mathbb{R}}$ (which is
monotone decreasing with respect to the first stochastic order) we have $%
\Phi (P_{X})\geq \Phi (P_{Y})$. It is also plausible that the initial price $%
y$ of $Y$ is not smaller than the one of $X$. However, if $y$ is
\textquotedblleft too large\textquotedblright\ compared to $x$ it is
possible that the corresponding P\&L distribution $T_{-y}P_{Y}$ is shifted
too much to the left, the two distributions $T_{-y}P_{Y}$ and $T_{-x}P_{X}$
intersect each other and the risk order is reverted: $\Phi
(T_{-x}P_{X})<\Phi (T_{-y}P_{Y})$.

\noindent For instance suppose that the distributions of $X$\ and $Y$ are
given by
\begin{eqnarray*}
F_{X}(z) &=&\left( 1\wedge (z+0,5)\right) \mathbf{1}_{[-0,5,+\infty )}(z),%
\text{ }z\in \mathbb{R}, \\
F_{Y}(z) &=&(1\wedge (z+0,5)^{2})\mathbf{1}_{[-0,5,+\infty )}(z),\text{ }%
z\in \mathbb{R},
\end{eqnarray*}%
and take $\Phi =V@R_{\lambda }$ with $\lambda =0.01$. Then $V@R_{\lambda
}(P_{X})=0.5-0.01=0,49>0.4=0.5-0.1=V@R_{\lambda }(P_{Y})$. But if $y>x+0.09$
then one easily checks that $V@R_{\lambda }(T_{-x}P_{X})<V@R_{\lambda
}(T_{-y}P_{Y})$.

\noindent If we focus only on payoffs, an agent is induced to prefer $Y$
respect to $X$ since $P_{X}\preccurlyeq _{1}P_{Y}$. But obviously in order
to hold position $Y$ the agent will have to pay an initial price which
influences the risk profile: the first stochastic dominance makes sense as
far as we compare positions having the same initial price.
\end{example}

We now provide the formal definition of Value\&Risk measures and their
properties.

\begin{definition}
\label{defi1} A \textbf{Value\&Risk} measure is any map $R:\mathbb{R}\times
\mathcal{L}(\Omega ,\mathcal{F})\times \mathcal{P}(\Omega )\rightarrow
\overline{\mathbb{R}}$ having the following four properties:

\begin{description}
\item[(1Mon)] for any fixed $(X,\mathbb{P})\in \mathcal{L}(\Omega ,\mathcal{F%
})\times \mathcal{P}(\Omega )$ and $p\leq q$ we have $R(p,X,\mathbb{P})\leq
R(q,X,\mathbb{P})$;

\item[(2Mon)] for any fixed $(p,\mathbb{P})\in \mathbb{R}\times \mathcal{P}%
(\Omega )$ and $X\leq Y$ we have $R(p,Y,\mathbb{P})\leq R(p,X,\mathbb{P})$;

\item[(3Mon)] for any fixed $(p,X)\in \mathbb{R}\times \mathcal{L}(\Omega ,%
\mathcal{F})$ and $\mathbb{P}^{1}\preccurlyeq _{X}\mathbb{P}^{2}$ we have $%
R(p,X,\mathbb{P}^{2})\leq R(p,X,\mathbb{P}^{1})$;

\item[(QCo)] Quasiconvex on $\mathcal{P}(\Omega )$: for any $p\in \mathbb{R}$%
, $X\in \mathcal{L}(\Omega ,\mathcal{F})$, $\mathbb{P}^{1},\mathbb{P}^{2}\in
\mathcal{P}(\Omega )$ and $\lambda \in (0,1)$ we have
\begin{equation*}
R(p,X,\lambda \mathbb{P}^{1}+(1-\lambda )\mathbb{P}^{2})\leq R(p,X,\mathbb{P}%
^{1})\vee R(p,X,\mathbb{P}^{2}).
\end{equation*}
\end{description}
\end{definition}

\noindent (1Mon) is simply justified by observing that the higher is the
price paid for $X,$ the higher is the risk. (2Mon) is the classical
monotonicity property for risk measures on random variables. (3Mon) and
(Qco) are the characteristic properties of risk measures on distributions
(see Definition \ref{defRM}). Proposition \ref{acceptance:mon} will
characterize these different types of monotonicity in terms of acceptance
sets.

\noindent The following condition is the appropriate extension, to this
context, of the law invariant property of risk measures:

\begin{description}
\item[(CLI)] Cross-Law Invariant: for any fixed $(X,\mathbb{P}^{1}),(Y,%
\mathbb{P}^{2})\in \mathcal{L}(\Omega ,\mathcal{F})\times \mathcal{P}(\Omega
)$ such that $P_{X}^{1}=P_{Y}^{2}$ then $R(p,X,\mathbb{P}^{1})=R(p,Y,\mathbb{%
P}^{2})$ for all $p\in \mathbb{R}$.
\end{description}

An additional feature (which in general fails in examples like $R(p,X;%
\mathbb{P})=p+V@R_{\lambda }(P_{X})$) is the quasiconvexity of the $R$ with
respect to the $X$ variable. This corresponds to the usual principle of
diversification as introduced in \cite{CMMMa}.

\begin{description}
\item[(QCoX)] Quasiconvex on $\mathcal{L}(\Omega ,\mathcal{F})$: for any $%
p\in \mathbb{R}$, $\mathbb{P} \in \mathcal{P}(\Omega )$, $X_1,X_2\in
\mathcal{L}(\Omega ,\mathcal{F})$ and $\lambda \in (0,1)$ we have
\begin{equation*}
R(p,\lambda X_1+(1-\lambda)X_2, \mathbb{P})\leq R(p,X_1,\mathbb{P} )\vee
R(p,X_2,\mathbb{P}).
\end{equation*}
\end{description}

\paragraph{V\&R measures and addition of cash.}

In Definition \ref{defi1} we do not require a priori any Cash Invariance
property of $R$. We now introduce the three axioms (Aff), (CA) and (DI) that
describe different level of invariancy with respect to additional cash and
needs to be studied separately. We will give a characterization of these
properties in Propositions \ref{acceptance:cash}.

\begin{definition}
\label{def2}Consider the following properties, with respect to addition of a
cash amount $\alpha \in \mathbb{R}$, that a Value\&Risk measure $R$ may
satisfy:

\begin{description}
\item[(Aff)] Price Affinity: $R(p+\alpha ,X,\mathbb{P})=R(p,X,\mathbb{P}%
)+\alpha $;

\item[(CA)] Cash Additivity: $R(p,X+\alpha ,\mathbb{P})=R(p,X,\mathbb{P}%
)-\alpha $;

\item[(DI)] Deviation Invariancy: $R(p+\alpha ,X+\alpha ,\mathbb{P})=R(p,X,%
\mathbb{P})$;

\item[(DCA)] Distribution Cash Additivity: $R(p,X,\mathbb{P}^{1})-\alpha
=R(p,Y,\mathbb{P}^{2})$ if $T_{\alpha }P_{X}^{1}=P_{Y}^{2}$.
\end{description}
\end{definition}

Finally we will also need the following property:

\begin{description}
\item[(Nor)] Normalization: $R(0,0,\mathbb{P})=0$ for all $\mathbb{P}\in
\mathcal{P}(\Omega )$;
\end{description}


\begin{remark}
\label{implications} Easy computations show that for a V\&R measure:

\begin{enumerate}
\item (Aff) and (CA) imply (DI); (Aff) and (DI) imply (CA); (CA) and (DI)
imply (Aff).

\item (DCA) iff (CA) and (CLI).

\item (Nor) and (DI) imply that $R(p,p,\mathbb{P})=0$ for any choice of $%
\mathbb{P}\in \mathcal{P}(\Omega )$.
\end{enumerate}
\end{remark}

\paragraph{Examples.}

\begin{itemize}
\item[(1)] First we consider the case in which $\Phi :\mathcal{P}(\mathbb{R}%
)\rightarrow \mathbb{R}\cup \left\{ \infty \right\} $ is a Risk Measure on
distribution, as in Definition \ref{defRM}, that also satisfies (TrI), as in
the case of the V@R or the Entropic Risk Measure. Define $R(p,X,\mathbb{P}%
):=\Phi (T_{-p}P_{X})$. By the property (TrI), coherently with equation (\ref%
{12}), we deduce
\begin{equation*}
R(p,X,\mathbb{P}):=\Phi (T_{-p}P_{X})=p+\Phi (P_{X}).
\end{equation*}%
Here the map $R(p,X,\mathbb{P})$ satisfies all the properties given in
Definitions \ref{defi1} and \ref{def2} and property (CLI) but not (Nor),
unless $\Phi (\delta _{0})=0$ ($\delta _{0}\in \mathcal{P}(\mathbb{R})$
being the Dirac distribution on $0\in \mathbb{R}$).

\item[(2)] In Appendix B we describe the risk measure $\Lambda V@R$,
introduced in \cite{FMP12}, which depends on a Probability/Loss function $%
\Lambda :\mathbb{R}\rightarrow \lbrack 0,1]$ and is defined as follows:
\begin{equation*}
\Lambda V@R(P_{X}):=-\sup \left\{ m\in \mathbb{R}\mid \mathbb{P}(X\leq
x)\leq \Lambda (x),\;\forall x\leq m\right\} .
\end{equation*}%
Define $R:\mathbb{R}\times \mathcal{L}(\Omega ,\mathcal{F})\times \mathcal{P}%
(\Omega )\rightarrow \overline{\mathbb{R}}$ by
\begin{eqnarray*}
R(p,X,\mathbb{P}) &=&\Lambda V@R(T_{-p}P_{X}) \\
&=&-\sup \left\{ m\in \mathbb{R}\mid P_{X}(-\infty ,y+p]\leq \Lambda
(y),\;\forall y\leq m\right\} .
\end{eqnarray*}%
By a simple change of variables and by defining the one parameter family $%
\Lambda ^{-p}$ as $\Lambda ^{-p}(x)=\Lambda (x-p)$ we get
\begin{equation}
R(p,X,\mathbb{P})=p+\Lambda ^{-p}V@R(P_{X})  \label{deco1}
\end{equation}%
Here the map $R$ satisfies (1-2-3Mon), (QCo), (DI), (CLI) but not (CA) nor
(Aff). Even though (CA) fails, we may deduce from equation (\ref{pha}) and (%
\ref{deco1}) that $R(p,p,\mathbb{P})=0$ independently from the choice of $%
\mathbb{P}$. We have $\Lambda (\cdot +\alpha )\geq \Lambda (\cdot )$ for $%
\alpha >0$, which implies $\Lambda ^{-p}V@R(P_{X+\alpha })=\Lambda
^{-p+\alpha }V@R(P_{X})-\alpha \geq \Lambda ^{-p}V@R(P_{X})-\alpha $ and
therefore
\begin{equation*}
\text{(Sup-CA)}\quad R(p,X+\alpha ,\mathbb{P})\geq R(p,X,\mathbb{P})-\alpha .
\end{equation*}%
Similarly if $\alpha <0$
\begin{equation*}
\text{(Sub-CA)}\quad R(p,X+\alpha ,\mathbb{P})\leq R(p,X,\mathbb{P})-\alpha .
\end{equation*}
\end{itemize}

\paragraph{Acceptance sets and Value\&Risk measures.\label{subAcceptance}}

We now consider a general family $\mathbb{A}=\{\mathcal{A}_{m}^{p}\}_{p,m\in
\mathbb{R}}$, $\mathcal{A}_{m}^{p}\subseteq \mathcal{P}(\mathbb{R})$ for
every $p,m\in \mathbb{R}$, and study the properties of the map
\begin{equation}
R_{\mathbb{A}}(p,X,\mathbb{P})=\inf \{m\mid P_{X}\in \mathcal{A}_{m}^{p}\}.
\label{capital}
\end{equation}%
As already mentioned the set $\mathcal{A}_{m}^{p}$ is intentionally
splitting the two components $p$ and $m$ corresponding respectively to the
\emph{initial value of the position} and the \emph{capital requirement to
cover expected losses}. \newline
We begin with the analysis of three different types of monotonicity and
quasiconvexity.

\begin{proposition}
\label{acceptance:mon} Consider a family $\{\mathcal{A}_{m}^{p}\}_{p,m\in
\mathbb{R}}$ contained in $\mathcal{P}(\mathbb{R})$ and $R_{\mathbb{A}}:%
\mathbb{R}\times \mathcal{L}(\Omega ,\mathcal{F})\times \mathcal{P}(\Omega
)\rightarrow \overline{\mathbb{R} }$ as defined in (\ref{capital}).

\begin{enumerate}
\item[(m1)] If for every $m\in \mathbb{R}$, $\mathcal{A}_{m}^{p}\subseteq
\mathcal{A}_{m}^{q}$ for $q\leq p$ then $R_{\mathbb{A}}$ is (1Mon).

\item[(m2)] If for every $(m,p,\mathbb{P})\in \mathbb{R}\times \mathbb{R}%
\times \mathcal{P}(\Omega )$, $X\leq Y$ and $P_{X}\in \mathcal{A}_{m}^{p}$
imply $P_{Y}\in \mathcal{A}_{m}^{p}$, then $R_{\mathbb{A}}$ is (2Mon).

\item[(m3)] If for every $(m,p,X)\in \mathbb{R}\times \mathbb{R}\times
\mathcal{L}(\Omega ,\mathcal{F})$, $\mathbb{P}_{1}\preccurlyeq _{X}\mathbb{P}%
_{2}$ and $P_{X}^{1}\in \mathcal{A}_{m}^{p}$ imply $P_{X}^{2}\in \mathcal{A}%
_{m}^{p}$, then $R_{\mathbb{A}}$ is (3Mon).

\item[(c)] Suppose that: (i) for all $p\in \mathbb{R}$ and all $\alpha \leq
\beta ,$ $\mathcal{A}_{\alpha }^{p}\subseteq \mathcal{A}_{\beta }^{p}$; (ii)
$\mathcal{A}_{m}^{p}$ is convex for all $p,m\in \mathbb{R}$; then $R_{%
\mathbb{A}}$ is (QCo).
\end{enumerate}

Viceversa take $R:\mathbb{R}\times \mathcal{L}(\Omega ,\mathcal{F})\times
\mathcal{P}\rightarrow \overline{R}$ and define $\mathbb{A}=\{\mathcal{A}%
_{m}^{p}\}_{p,m}$ by: $\mathcal{A}_{m}^{p}=\{Q_{X}\in \mathcal{P}(\mathbb{R}%
)\mid R(p,X,\mathbb{Q})\leq m\}$. Then:

\begin{enumerate}
\item[(M1)] If $R$ is (1Mon) then $\mathcal{A}_{m}^{p}\subseteq \mathcal{A}%
_{m}^{q}$ for $q\leq p$ and $m\in \mathbb{R}$.

\item[(M2)] If $R$ is (2Mon) then for every $(m,p,\mathbb{P})\in \mathbb{R}%
\times \mathbb{R}\times \mathcal{P}(\Omega )$, $X\leq Y$ and $P_{X}\in
\mathcal{A}_{m}^{p}$ imply $P_{Y}\in \mathcal{A}_{m}^{p}$.

\item[(M3)] If $R$ is (3-Mon) then for every $(m,p,X)\in \mathbb{R}\times
\mathbb{R}\times \mathcal{L}(\Omega ,\mathcal{F})$, $\mathbb{P}%
^{1}\preccurlyeq _{X}\mathbb{P}^{2}$ and $P_{X}^{1}\in \mathcal{A}_{m}^{p}$
imply $P_{X}^{2}\in \mathcal{A}_{m}^{p}$.

\item[(C)] If $R$ is (QCo) then $\mathcal{A}_{m}^{p}$ is convex for every $%
p,m\in \mathbb{R}$.
\end{enumerate}
\end{proposition}

\begin{proof}
(m1) Let $q\leq p$, by assumption $\{m\in \mathbb{R}\mid P\in \mathcal{A}%
_{m}^{p}\}\subseteq \{m\in \mathbb{R}\mid P\in \mathcal{A}_{m}^{q}\}$ .
Hence $R_{\mathbb{A}}(p,X,\mathbb{P})=\inf \{m\in \mathbb{R}\mid
T_{m}P_{X}\in \mathcal{A}^{p}\}\geq \inf \{m\in \mathbb{R}\mid T_{m}P_{X}\in
\mathcal{A}^{q}\}=R_{\mathbb{A}}(q,X,\mathbb{P})$. Similarly for (m2) and
(m3). \newline
(c) We fix $(p,X)\in \mathbb{R}\times \mathcal{L}(\Omega ,\mathcal{F})$ and
consider the map $R_{\mathbb{A}}(p,X,\cdot )$. We want to show that the
sublevels of this map are convex. Let $B_{a}=\{\mathbb{Q}\in \mathcal{P}%
(\Omega )\mid R_{\mathbb{A}}(p,X,\mathbb{Q})\leq a\}$ and $\mathbb{P}^{1},%
\mathbb{P}^{2}\in B_{a}$. Assume w.l.o.g. that $a\geq M:=R_{\mathbb{A}}(p,X,%
\mathbb{P}^{1})\geq R_{\mathbb{A}}(p,X,\mathbb{P}^{2})$. Fix any $%
\varepsilon >0.$ Then there exist $M^{i}\leq M+\varepsilon $ $(i=1,2)$ such
that $P_{X}^{i}\in \mathcal{A}_{M^{i}}^{p}$. Since $\mathcal{A}%
_{M^{i}}^{p}\subseteq \mathcal{A}_{M+\varepsilon }^{p}$ and $\mathcal{A}%
_{M+\varepsilon }^{p}$ is convex, we deduce that $\lambda
P_{X}^{1}+(1-\lambda )P_{X}^{2}\in \mathcal{A}_{M+\varepsilon }^{p}$ for any
$\lambda \in (0,1).$ Then $R_{\mathbb{A}}(p,X,\lambda \mathbb{P}%
^{1}+(1-\lambda )\mathbb{P}^{2})=\inf \{m\mid \lambda P_{X}^{1}+(1-\lambda
)P_{X}^{2}\in \mathcal{A}_{m}^{p}\}\leq M+\varepsilon $. As this holds for
any $\varepsilon >0$ we obtain $R_{\mathbb{A}}(p,X,\lambda \mathbb{P}%
^{1}+(1-\lambda )\mathbb{P}^{2})\leq M\leq a$. \newline
Items (M1-2-3) and (C) are straightforward consequences of the definitions.
\end{proof}

\bigskip

We are interested in possible declinations of the family $\mathbb{A}=\{%
\mathcal{A}_{m}^{p}\}_{p,m\in \mathbb{R}}$ which leads to different types of
behavior with respect to cash addition. The following Proposition fully
characterizes those Cross-Law-Invariant maps $R$ that satisfy either (CA) or
(Aff) or (DI).

\begin{proposition}
\label{acceptance:cash} Consider a family $\{\mathcal{A}_{m}^{p}\}_{p,m\in
\mathbb{R}}$ contained in $\mathcal{P}(\mathbb{R})$ and $R_{\mathbb{A}}:%
\mathbb{R}\times \mathcal{L}(\Omega ,\mathcal{F})\times \mathcal{P}(\Omega
)\rightarrow \overline{\mathbb{R} }$ as defined in (\ref{capital}).

\begin{enumerate}
\item[(CA)] If $\mathcal{A}_{m}^{p}=T_{-m}\mathcal{A}^{p}=\{T_{-m}P\mid P\in
\mathcal{A}^{p}\}$ , for a given family $\{\mathcal{A}^{p}\}_{p\in \mathbb{R}%
}\subseteq \mathcal{P}(\mathbb{R})$, then
\begin{equation*}
R_{\mathbb{A}}(p,X,\mathbb{P})=\inf \{m\mid T_{m}P_{X}\in \mathcal{A}^{p}\},
\end{equation*}%
and $R_{\mathbb{A}}$ is (DCA) and hence (CA) and (CLI). \newline
Viceversa take $R:\mathbb{R}\times \mathcal{L}(\Omega ,\mathcal{F})\times
\mathcal{P}(\Omega )\rightarrow \overline{\mathbb{R}}$ satisfying (DCA).
Define $\mathbb{A}=\{T_{-m}\mathcal{A}^{p}\}_{p,m\in \mathbb{R}}$ where $%
\mathcal{A}^{p}=\{Q_{X}\in \mathcal{P}(\mathbb{R})\mid X\in \mathcal{L}%
(\Omega ,\mathcal{F})$ and $R(p,X,\mathbb{Q})\leq 0\}$. Then $R_{\mathbb{A}%
}=R$.

\item[(Aff)] If $\mathcal{A}_{m}^{p}$ satisfies for every $\alpha $, $%
\mathcal{A}_{m}^{p+\alpha }=\mathcal{A}_{m-\alpha }^{p}$ then there exists $%
\beta :\mathcal{P}(\mathbb{R})\rightarrow \overline{\mathbb{R}}$ such that $%
R_{\mathbb{A}}(p,X,\mathbb{P})=p+\beta (P_{X})$ and $R_{\mathbb{A}}$ is
(Aff) and (CLI). \newline
Viceversa take $R:\mathbb{R}\times \mathcal{P}\rightarrow \overline{R}$
satisfying (Aff) and (CLI) and define $\mathbb{A}=\{\mathcal{A}%
_{m}^{p}\}_{p,m}$ where $\mathcal{A}_{m}^{p}=\{Q_{X}\in \mathcal{P}(\mathbb{R%
})\mid X\in \mathcal{L}(\Omega ,\mathcal{F})$ and $R(p,X,\mathbb{Q})\leq m\}$%
. Then $\mathcal{A}_{m}^{p+\alpha }=\mathcal{A}_{m-\alpha }^{p}$ for all $%
\alpha $ and $R_{\mathbb{A}}=R$.

\item[(DI)] If $\mathcal{A}_{m}^{p}=\{Q\in \mathcal{P}(\mathbb{R})\mid
T_{-p}Q\in \mathcal{A}_{m}^{0}\}$ then there exists $\beta :\mathcal{P}(%
\mathbb{R})\rightarrow \overline{\mathbb{R}}$ such that $R_{\mathbb{A}}(p,X,%
\mathbb{P})=\beta (T_{-p}P_{X})$ and $R_{\mathbb{A}}$ is (DI) and (CLI).
\newline
Viceversa take $R:\mathbb{R}\times \mathcal{P}\rightarrow \overline{R}$
satisfying (DI) and (CLI) and define $\mathbb{A}=\{\mathcal{A}%
_{m}^{p}\}_{p,m}$ where $\mathcal{A}_{m}^{0}=\{Q_{X}\in \mathcal{P}(\mathbb{R%
})\mid R(0,X,\mathbb{Q})\leq m\}$ and $\mathcal{A}_{m}^{p}=\{Q\in \mathcal{P}%
(\mathbb{R})\mid T_{-p}Q\in \mathcal{A}_{m}^{0}\}$. Then $R_{\mathbb{A}}=R$.

\end{enumerate}
\end{proposition}

\begin{proof}
(CA): the first implication is straightforward. For the viceversa notice
that $R_{\mathbb{A}}(p,X,\mathbb{P})=\inf\{m\mid P_X\in T_{-m}\mathcal{A}%
^p\}= \inf\{m\mid T_mP_X\in \mathcal{A}^p\}= \inf\{m\mid R(p,X+m,\mathbb{P}%
)\leq 0\}= \inf\{m\mid R(p,X,\mathbb{P})\leq m\}=R(p,X,\mathbb{P})$. \newline
(Aff): we show the existence of $\beta$ by observing $R_{\mathbb{A}}(p,X,%
\mathbb{P})=\inf\{m\mid P_X\in \mathcal{A}^p_m\}=p+\inf\{m-p\mid P_X\in
\mathcal{A}^p_m\}=p+\inf\{m\mid P_X\in \mathcal{A}^p_{m+p}\}=p+\inf\{m\mid
P_X\in \mathcal{A}^0_{m}\}$ and therefore $\beta(P_X):=\inf\{m\mid P_X\in
\mathcal{A}^0_{m}\}$. The viceversa is similar to the case (CA). \newline
(DI): both implications follows as in the previous cases.
\end{proof}

\begin{remark}
\label{delbaen}A particular case of (DCA) is when the set $\mathcal{A}^{p}$
in $\mathcal{A}_{m}^{p}:=T_{-m}\mathcal{A}^{p}$ is independent from $p$ (in
which case we may set $\mathcal{A}_{m}^{p}=T_{-m}\mathcal{A}^{0}$ ). Then $%
R_{\mathbb{A}}(p,X,\mathbb{P})=$ $R_{\mathbb{A}}(0,X,\mathbb{P})$, for any $%
p\in \mathbb{R}$, and this corresponds, as mentioned in the Introduction, to
the intuition proposed in the original paper by Delbaen et al. \cite{ADEHb}
that bygones are bygones.
\end{remark}


%

\section{Model Risk and Intrinsic Risk}

In Section \ref{interpretation:quasiconvex} (resp. \ref{secVR}) we develop
the two approaches sketched in the Introduction. In Section \ref{secVR} we
will introduce intrinsic risk maps which constitute particular $V$\&$R$
measures.

\subsection{Use Models to Test Claims \label{interpretation:quasiconvex}}

Here we adopt the Knightian uncertainty point of view. We consider a set $%
\mathcal{M}\subseteq \mathcal{P}(\Omega )$ of probability measures $\mathbb{Q%
}$ on $(\Omega ,\mathcal{F})$, each representing a possible pricing rule,
and for a given $X\in \mathcal{L}(\Omega ,\mathcal{F})$ the corresponding set%
\begin{equation*}
\mathcal{M}_{X}:=\left\{ Q_{X}\in \mathcal{P(\mathcal{B}}(\mathbb{R})%
\mathcal{)}\mid Q_{X}=\mathbb{Q}(X\leq x)\text{, }\mathbb{Q}\in \mathcal{M}%
\right\}
\end{equation*}%
of associated probability distribution $Q_{X}$ on $(\mathbb{R},\mathcal{B}(%
\mathbb{R}))$. For example, $\mathcal{M}$ could be a set of calibrated
martingale measures, i.e. those induced by a fixed set of benchmark
contingent claims $F^{i}$ each having initial cost $F_{0}^{i}$%
\begin{equation*}
\mathcal{M}=\{\mathbb{Q}\in \mathcal{P}(\Omega )\mid E_{\mathbb{Q}}[S]=S_{0}%
\text{ and }E_{\mathbb{Q}}[F^{i}]=F_{0}^{i}\text{, }i=1,...,N\}.
\end{equation*}

In this approach we assume that we have a criterion to asses the correctness
of our selection, by assuming the existence of a model-risk function $A:%
\mathcal{P}(\Omega )\times \mathcal{L}(\Omega ,\mathcal{F})\rightarrow
\overline{\mathbb{R}}$ which asses the risk (or level of ambiguity) in the
choice of a probability $\mathbb{Q}$, whenever we are modelling a random
payoff $X$. A small value of $A(\mathbb{Q},X)$ means that we are quite
confident in our choice. We proceed in four steps.

\begin{itemize}
\item[(i)] We \textquotedblleft test\textquotedblright\ the claim $f(X)$
over the set $\mathcal{M}$ under the constraint $A(\mathbb{Q},X)\leq a$ and
obtain a Value function.
\end{itemize}

\begin{definition}
\label{defPrimal}Let $X\in \mathcal{L}(\Omega ,\mathcal{F}),$ $\mathcal{M}%
\subseteq \mathcal{P}(\Omega )$, $a\in \mathbb{R}$ and $A:\mathcal{P}(\Omega
)\times \mathcal{L}(\Omega ,\mathcal{F})\rightarrow \overline{\mathbb{R}}$.
Define the map $V_{A}:\mathbb{R}\times \mathcal{L}(\Omega ,\mathcal{F}%
)\times C_{b}\rightarrow \overline{\mathbb{R}}$%
\begin{equation*}
V_{A}(a,X;f):=\sup_{\mathbb{Q}\in \mathcal{M}}\left\{ E_{\mathbb{Q}}\left[
f(X)\right] \mid A(\mathbb{Q},X)\leq a\right\} =\sup_{\mathbb{Q}\in \mathcal{%
M}}\left\{ \int f(X)d\mathbb{Q}\mid A(\mathbb{Q},X)\leq a\right\} .
\end{equation*}
\end{definition}

\noindent\textit{Here }$a$\textit{\ and }$X$\textit{\ are given and we test
the price of the claim }$f(X)$\textit{\ over the set }$\mathcal{M}$ (compare
with Definition \ref{defdual}). We will omit the dependence of $V_{A}$ from $%
A$, whenever no confusion may arise.

\begin{remark}
\label{remPrimal}(compare with Remark \ref{remdual} and equation (\ref{123}%
)). In many cases there might exist $\widetilde{A}:\mathcal{M}%
_{X}\rightarrow \mathbb{R}\cup \left\{ \infty \right\} $ such that $A(%
\mathbb{Q},X)=\widetilde{A}(\mathbb{Q}\circ X^{-1}).$ If this is the case we
reduce the problem to
\begin{equation}
V(a,X;f)=\sup_{Q\in \mathcal{M}_{X}}\left\{ \int fdQ\mid \widetilde{A}%
(Q)\leq a\right\} .  \label{V}
\end{equation}
\end{remark}

\noindent The value
\begin{equation*}
v=V(a,X;f)
\end{equation*}%
is the maximum value of the contingent claim $f$ on the underlying $X,$ for
the level $a$ of model-risk in the choice of $Q$, i.e. is the best (seller)
price of the claim $f(X)$ relative to all possible choices of pricing
measures under the constraint that the model risk is less than or equal to $%
a $.

\begin{itemize}
\item[(ii)] The Intrinsic Model Risk.
\end{itemize}

\noindent By defining the generalized inverse of $V(\cdot ,X;f)$ we obtain:%
\begin{equation*}
a=V^{-1}(v,X;f)=\inf \{s\in \mathbb{R}\mid V(s,X;f)\geq v\}
\end{equation*}%
which represents the minimum model risk one has to accept relative to a
claim $f$ on $X,$ having price larger than $v.$ In other words, $a$ is the
smallest model risk the decision maker is forced to accept in order to find
a pricing model that attributes to $f(X)$ the price $v$.

\begin{itemize}
\item[(iii)] The Indirect Model-Risk function.
\end{itemize}

If a pricing model $\mathbb{Q\in }\mathcal{P}(\Omega )$ is determined, the
quantity%
\begin{equation*}
V^{-1}\left( \int fdQ_{X},X;f\right) =\inf \left\{ s\in \mathbb{R}\mid
V(s,X;f)\geq \int fdQ_{X}\right\}
\end{equation*}%
is the risk associated to the choice of the distribution $Q_{X}$, induced by
$X$, for pricing the particular claim $f$. Let $K\subseteq C_{b}$ and let%
\begin{equation*}
\alpha _{K}(Q_{X}):=\sup_{f\in K}V^{-1}\left( \int fdQ_{X},X;f\right)
\end{equation*}%
be the maximum (w.r.to $f\in K)$ model risk associated to $\mathbb{Q}$,
given the underlying $X.$ We then see that starting from the a priori given
model risk function $\widetilde{A}$ we end up with another map $\alpha _{K}:%
\mathcal{M}_{X}\rightarrow \overline{\mathbb{R}}$\ induced by $\widetilde{A}$
and $K$, which can be interpreted as the \textquotedblleft Indirect Model
Risk\textquotedblright\ function. \newline

\begin{itemize}
\item[(iv)] Duality.
\end{itemize}

The natural problem now is to find conditions on the set $K$ for which $%
\alpha _{K}=\widetilde{A}.$ The solution is given by the following result,
a\ reformulation of Proposition \ref{propvolle} in Appendix B.

\begin{proposition}
\label{propA}Let $\widetilde{A}:\mathcal{P}(\mathbb{R})\rightarrow \overline{%
\mathbb{R}}$ be quasi-convex, $\preccurlyeq _{1}$-monotone decreasing and $%
\sigma (\mathcal{P}(\mathbb{R}),C_{b}(\mathbb{R}))$-lsc and let $X\in
\mathcal{L}(\Omega ,\mathcal{F}).$ Then
\begin{equation*}
\widetilde{A}(Q_{X})=\sup_{f\in C_{b}^{-}}V_{\widetilde{A}}^{-1}\left( \int
fdQ_{X},X;f\right) =\alpha _{C_{b}^{-}}(Q_{X}).
\end{equation*}
\end{proposition}

This also shows that whenever $K\subseteq C_{b}^{-}$ then the indirect model
risk function is less conservative than $\widetilde{A},$ i.e. $\alpha
_{K}\leq \widetilde{A}$.

\subsection{Use Claims to Test Models\label{secVR}}

In this section we explain our approach that constitutes one of the main
contributions of this paper. It can be considered as the dual formulation of
the situation described in Section \ref{interpretation:quasiconvex} and the
presentation will intentionally follow the analogous four steps of the
previous section. Given a position $X\in \mathcal{L}(\Omega ,\mathcal{F})$
and a probability $\mathbb{Q}\in \mathcal{P}$, we look at all possible
prices $E_{\mathbb{Q}}[f(X)],$ for $f$ belonging to a subset $K$ of $%
C^{0}:=C^{0}(\mathbb{R})$ the space of continuous functions on $\mathbb{R}$.
The idea is to use the claims in $K$, or in the set
\begin{equation*}
K_{X}:=\left\{ F\in \mathcal{L}(\Omega ,\mathcal{F})\mid F=f(X)\text{, }f\in
K\right\} ,
\end{equation*}%
to test the pricing rule $\mathbb{Q}$.

\bigskip

In this approach we assume the existence of a map $\varphi :K\times \mathcal{%
L}(\Omega ,\mathcal{F})\rightarrow \overline{\mathbb{R}}$, where
$\varphi (f,X)$ assigns the \textit{risk reduction} the agent will
benefit by introducing a derivative $f(X)$ to cover the losses of
$X$. Such function $\varphi $ has the analogue role of the map $A$ introduced in Section \ref%
{interpretation:quasiconvex}, but a different interpretation (see
the examples below). Indeed $\varphi (f,X)$ will determine all
claims $f\in K$ having at
most the same level of risk reduction and use these claims to test $\mathbb{Q%
}$. The nomenclature `risk reduction' relates to the fact that we are
looking to the effects that additional derivatives have on the overall risk.

\begin{example}
\label{primoesempio} Some examples can be easily built up considering a
classical risk measure $\rho :\mathcal{L}(\Omega ,\mathcal{F})\rightarrow
\overline{\mathbb{R}}$, namely:
\begin{eqnarray*}
\varphi (f,X) &=&-\rho (f(X));\quad \quad \quad \quad \quad \varphi
(f,X)=\rho (X-f(X)); \\
\varphi (f,X) &=&\rho (X)-\rho (f(X));\quad \quad \varphi (f,X)=-\rho
(f(X)-X).
\end{eqnarray*}
\end{example}

In this examples the choice of a specific risk measure could be strongly
related to the knowledge of a reference probability if $\rho =\rho _{\mathbb{%
P}}$. If we do not want to rely on $\mathbb{P}$ nor on $X$ natural choices
for $\varphi $ are:
\begin{eqnarray}
\varphi (f,X) &=&-\inf_{x\in \mathbb{R}}\{x-f(x)\}=\sup_{x\in \mathbb{R}%
}\{f(x)-x\}  \label{1} \\
\varphi (f,X) &=&\inf_{x\in \mathbb{R}}\{f(x)-x\}  \label{2} \\
\varphi (f,X) &=&\inf_{x\in \mathbb{R}}\{f(x)\}  \label{3}
\end{eqnarray}%
which assign the risk reduction led by selling/buying $f$ in the worst case
scenario, whatever underlying we are considering and independently from the
reference probability $\mathbb{P}$. The risk reductions defined in (\ref{1})
(\ref{2}) and (\ref{3}) satisfy the following condition:
\begin{equation}
\varphi (f,X)=\varphi (f,Y),\text{ for all }f\in K\text{ and }X,Y\in
\mathcal{L}(\Omega ,\mathcal{F}).  \label{equal1}
\end{equation}%
Since $K$ could be very small, (\ref{equal1}) may be weaker than requiring
that $\varphi (f,X)$ is independent from $X$. We will explain better this
fact in the examples provided in Section \ref{examples}.

Similarly to the previous section, we proceed in four steps.

\begin{itemize}
\item[(i)] We test $Q_{X}$ over the set $K$ under the constraint $\varphi
(f,X)\geq r$ and obtain a Price function.
\end{itemize}

\begin{definition}
\label{defdual}Let $X\in \mathcal{L}(\Omega ,\mathcal{F})$, $K\subseteq
C^{0},$ $r\in \mathbb{R}$ and $\varphi :\mathcal{L}(\Omega ,\mathcal{F}%
)\rightarrow \overline{\mathbb{R}}$ be a risk reduction. Define the map $\Pi
_{\varphi }:\mathbb{R}\times \mathcal{L}(\Omega ,\mathcal{F})\times \mathcal{%
P}(\Omega )\rightarrow \overline{\mathbb{R}}$
\begin{equation}
\Pi _{\varphi }(r,X;\mathbb{Q}):=\inf_{f\in K}\left\{ E_{\mathbb{Q}}\left[
f(X)\right] \mid \varphi (f,X)\geq r\right\} =\inf_{f\in K}\left\{ \int
fdQ_{X}\mid \varphi (f,X)\geq r\right\}  \label{Pi}
\end{equation}
\end{definition}

\noindent \textit{Here }$r$\textit{\ and }$X$\textit{\ are given and we test
}$Q_{X}$\textit{\ over the set }$K$ (compare with Definition \ref{defPrimal}%
).

\begin{remark}
\label{remdual}(Compare with Remark \ref{remPrimal}) In some cases, for
example when $\varphi (f,X)=-\rho (f(X))$, the function $\varphi $ can be
written as $\varphi (f,X)=\widetilde{\varphi }(f(X))$ with $\widetilde{%
\varphi }:K_{X}\rightarrow \overline{\mathbb{R}}$. In this case
\begin{equation*}
\Pi_{\varphi} (r,X;\mathbb{Q})=\inf_{F\in K_{X}}\left\{ \int
Fd\mathbb{Q}\mid \widetilde{\varphi }(F)\geq r\right\} .
\end{equation*}
\end{remark}

The price
\begin{equation}
p=\Pi_{\varphi} (r,X;\mathbb{Q})  \label{equality:price}
\end{equation}%
corresponds to the cheapest $\mathbb{Q}-$price of any derivative
(on $X$) in the class $K$ which guarantees a reduction (at least
equal to $r$) of the level of risk. When $p=\Pi_{\varphi}
(r,X;\mathbb{Q})$ then there exists some $f(X)\in K_{X}$ having
$\mathbb{Q}$-price almost equal to $p$ and risk reduction at
least equal to $r$. Notice that from the buyer point of view, an underlying $%
X$ bought at price $p_{2}$ is riskier than the same $X$ bought at the
smaller price $p_{1}<p_{2}$.

\begin{itemize}
\item[(ii)] The Intrinsic Risk:
\end{itemize}

\begin{definition}
\label{defR}The map $R_{\varphi}:\mathbb{R}\times \mathcal{L}(\Omega ,\mathcal{F}%
)\times \mathcal{P}(\Omega )\rightarrow \overline{\mathbb{R}}$ is
the
generalized left inverse of $\Pi_{\varphi} (\cdot ,X;\mathbb{Q})$:%
\begin{equation}
R_{\varphi}(p,X;\mathbb{Q}):=\Pi_{\varphi}
^{-1}(p,X;\mathbb{Q})=:\sup \{s\in \mathbb{R}\mid \Pi_{\varphi}
(s,X;\mathbb{Q})\leq p\},  \label{R}
\end{equation}%
with $R_{\varphi}(p,X;\mathbb{Q})=-\infty $ for $p<\inf_{f\in
K}E_{\mathbb{Q}}[f(X)]$ and $R_{\varphi}(p,X;\mathbb{Q})=+\infty $
for $p>\sup_{f\in K}E_{\mathbb{Q}}[f(X)]$.
\end{definition}

 We will omit the dependence of $\Pi _{\varphi}, R_{\varphi}$
from $\varphi $, whenever no confusion may arise.

Suppose that we buy $X$ at a price $p$ and that $\mathbb{Q}$ is the correct
pricing rule. We set: $r:=\sup \{s\in \mathbb{R}\mid \Pi (s,X;\mathbb{Q}%
)\leq p\}$. Then $r$ is the maximal risk reduction $s$ for which $\Pi (s,X;%
\mathbb{Q})\leq p$, i.e. the maximal risk reduction that allows to find a
claim $f(X)$ having $\mathbb{Q}$-price not larger than $p$ and risk
reduction at least equal to $r$. Therefore $r=R(p,X;\mathbb{Q})$ is the
intrinsic risk of acquiring $X$ at price $p$, assuming that $Q$ is the
correct pricing rule, and\emph{\ }it corresponds, for particular functions%
\emph{\ }$\varphi (f,X)$\emph{,} to the maximal risk reduction we would
obtain buying $X$ and selling a derivative $f(X)$ with a price at most equal
to $p$ (see examples in sections \ref{esempio:1} and \ref{esempio:tail}).

\begin{remark}
\label{reduction}The previous interpretation can be more precisely explained
as follows. The maximal risk reduction an agent may obtain by selling
derivatives (on $X)$ with $\mathbb{Q}$-price smaller than $p$ is given by
the function
\begin{equation*}
H(p,X;\mathbb{Q}):=\sup_{F\in K_{X}}\left\{ \widetilde{\varphi }(F)\mid E_{%
\mathbb{Q}}[F]\leq p\right\} ,
\end{equation*}%
with $\widetilde{\varphi }$ given in Remark \ref{remdual}. Then Proposition %
\ref{propHH} will show that $R(\cdot ,X;\mathbb{Q})$ is the right-continuous
version of $H(\cdot ,X;\mathbb{Q})$.
\end{remark}

\begin{example}[A simple case.]
If we assume that $K=C_{+}^{0}=\{f\in C^{0}\mid f\text{ increasing }\}$ and $%
\varphi (f(X))=\inf_{x\in \mathbb{R}}(f(x)-x)$ then $\Pi (r,X,\mathbb{Q})=E_{%
\mathbb{Q}}[X]+r,$ with the abuse of notation that if $E_{\mathbb{Q}}[X]=\pm
\infty $ then $\Pi (r,X,\mathbb{Q})=\pm \infty $. Thus
\begin{equation*}
R(p,X;\mathbb{Q})=p-E_{\mathbb{Q}}[X]
\end{equation*}%
and, as already mentioned, in this case the intrinsic risk $R(p,X,\mathbb{Q}%
) $ of acquiring $X$ at price $p$ is exactly the discrepancy between $p$ and
$E_{\mathbb{Q}}[X]$, assuming that $\mathbb{Q}$ is the pricing rule.
\end{example}

\begin{itemize}
\item[(iii)] The Indirect Risk Reduction
\end{itemize}

If a contingent claim $f(X)$ is determined, then the quantity%
\begin{eqnarray*}
R(E_{\mathbb{Q}}\left[ f(X)\right] ,X;\mathbb{Q})& := & \Pi ^{-1}(E_{\mathbb{Q}}%
\left[ f(X)\right] ,X;\mathbb{Q})
\\ & = & \sup \{s\in \mathbb{R}\mid \Pi (s,X;%
\mathbb{Q})\leq E_{\mathbb{Q}}\left[ f(X)\right] \}
\end{eqnarray*}
is the risk reduction we face buying $X$ at the price $E_{\mathbb{Q}}\left[
f(X)\right] $. If $\mathcal{M}\subseteq ca(\Omega )$ then
\begin{equation*}
\Psi _{\mathcal{M}}(f(X)):=\inf_{\mathbb{Q}\in \mathcal{M}}R(E_{\mathbb{Q}%
}[f(X)],X;\mathbb{Q})
\end{equation*}%
represents the smallest (with respect to all $\mathbb{Q}\in \mathcal{M}$)
risk reduction associated to the claim $f(X)$.

\begin{itemize}
\item[(iv)] Duality.
\end{itemize}

Such indirect risk reduction\ function $\Psi _{\mathcal{M}}$ should then be
compared with the one we started from. The following proposition, an
immediate consequence of Theorem \ref{thA} in Appendix, shows under which
conditions we might recover $\varphi $ from $\Psi _{\mathcal{M}}$.

\begin{proposition}
\label{dualityK}Assume that $K_{X}\subset \mathcal{L}^{\infty }(\Omega ,%
\mathcal{F})$ is a convex cone, $\sigma (\mathcal{L}^{\infty },ca)$-closed.
If $\varphi :K_{X}\rightarrow \overline{\mathbb{R}}$ is monotone increasing,
quasiconcave and $\sigma (\mathcal{L}^{\infty },ca)$-upper semicontinuous.
Then
\begin{equation*}
\varphi (f(X))=\inf_{\mathbb{Q}\in \mathcal{K}_{1}^{\circ }}R_{\varphi }(E_{%
\mathbb{Q}}[f(X)],X;\mathbb{Q})=\Psi _{\mathcal{K}_{1}^{\circ }}(f(X)),
\end{equation*}%
where $\mathcal{K}_{1}^{\circ }=\{\mu \in ca\mid \mu (\Omega )=1\text{ and }%
\int f(X)d\mu \geq 0\;\forall f\in K\text{ s.t. }f(X)\geq 0\}$. (Compare
with Proposition \ref{propA}).
\end{proposition}

The following theorem, a consequence of Proposition \ref{properties} in
Section \ref{GeneralResults}, shows that the maps $R$ defined in (\ref{R})
are V\&R measures.

\begin{theorem}
\label{Price&Risk}Suppose that $K\subseteq C_{+}^{0}(\mathbb{R})$ and that $%
\varphi :K\times \mathcal{L}(\Omega ,\mathcal{F})\rightarrow \overline{%
\mathbb{R}}$ satisfies (\ref{equal1}). Then $R_{\varphi}$ defined
by (\ref{R}) is a Value\&Risk measure that satisfies
(\textbf{CLI}).

If in addition $K\subseteq \left\{ f\in C_{+}^{0}(\mathbb{R})\mid f\text{
concave}\right\} $ then $R_{\varphi}$ defined by (\ref{R}) satisfies also (\textbf{QCoX%
}).
\end{theorem}

It is worth mentioning that in virtue of Proposition
\ref{properties} (a1) and (a2) the properties (1Mon) and (QCo) for
the map $R$ will hold independently from the properties of $K$ or
$\varphi $. The following Proposition (which proof is postponed in
Appendix A) considers a fixed couple $(X,\mathbb{Q})$ and studies
three properties of $\Pi $ and $R$ with respect to monotonicity,
convex combinations and Minkowski sum of the sets $K$ of testing
claims.

\begin{proposition}\label{dependence:K}
For fixed $X\in \mathcal{L}(\Omega ,\mathcal{F})$ and $Q\in \mathcal{P}%
(\Omega )$, consider $K\subseteq C^{0}$ and denote
\begin{eqnarray*}
\Pi _{K}(r) &:&=\inf_{f\in K}\left\{ E_{\mathbb{Q}}\left[ f(X)\right] \mid
\varphi (f,X)\geq r\right\} , \\
R_{K}(p) &:&=\sup \left\{ s\in \mathbb{R}\mid \Pi _{K}(s)\leq p\right\} .
\end{eqnarray*}%
Let $K^{1},K^{1}\subseteq C^{0}$.

\begin{itemize}
\item[1] If $K^{1}\subseteq K^{2}$ then $\Pi _{K^{1}}\geq \Pi _{K^{2}}$ and $%
R_{K^{1}}\leq R_{K^{2}}$.

\item[2] If $\varphi $ is quasiconcave as a function of $f$ then for any
fixed $\lambda \in \lbrack 0,1]$
\begin{eqnarray}
\Pi _{K^{\lambda }} &\leq &\lambda \Pi _{K^{1}}+(1-\lambda )\Pi _{K^{2}},
\label{PIqco} \\
R_{K^{\lambda }} &\geq &R_{K^{1}}\wedge R_{K^{2}}  \label{Rqco}
\end{eqnarray}%
where
\begin{eqnarray*}
K^{\lambda } & := & \lambda K^{1}+(1-\lambda )K^{2} \\ & =
&\{f^{\lambda }\in C^{0}\mid f^{\lambda}\lambda f^{1}+(1-\lambda
)f^{2}\text{ for }f^{1}\in K^{1}\text{ and }f^{2}\in K^{2}\}.
\end{eqnarray*}

\item[3] If, for each $r\in \mathbb{R}$, $\varphi (f^{1},X)\geq r$ and $%
\varphi (f^{2},X)\geq r$ implies $\varphi (f^{1}+f^{2},X)\geq r$ then:%
\begin{eqnarray*}
\Pi _{K^{1}+K^{2}} &\leq &\Pi _{K^{1}}+\Pi _{K^{2}}, \\
R_{K^{1}+K^{2}}(p) &\geq &\sup_{p_1+p_2=p}\left\{
R_{K^{1}}(p_1)\wedge R_{K^{2}}(p_2)\right\} \text{, }p\in
\mathbb{R}\text{.}
\end{eqnarray*}
\end{itemize}
\end{proposition}

It is clear that for a larger set of testing claims the price $\Pi $ will
decrease and the risk reduction $R$ will increase. From Item 2, we deduce
that by taking convex combinations of two sets $K^{1}$ and $K^{2}$ the risk
reduction $R_{K^{\lambda }}$ is always larger than the minimum $%
R_{K^{1}}\wedge R_{K^{2}}$ of the two single risk reductions, so that the
operation of taking convex combination is encouraged.

\section{Examples of V\&R measures from Def. \protect\ref{defR}}

\label{examples}

\subsection{Control of unbounded losses of the underlying by releasing
options.}

\label{esempio:1} In this first example we consider the case in which the
underlying $X\in \mathcal{L} (\Omega ,\mathcal{F})$ produces a potentially
unbounded loss, in particular $\inf_{\omega\in\Omega} X(\omega)=-\infty$. We
here consider a fairly general class of approximating test functions
described by a family $\{f_{\alpha}\}_{\alpha\in \mathbb{R} }$ such that

\begin{enumerate}
\item $\{f_{\alpha}\}_{\alpha\in \mathbb{R} }\subset C^0_+$ and $f_{0}(x)=x$
for every $x\in\mathbb{R} $;

\item $f_{\alpha}(x)< f_{\beta}(x)$ for every $\alpha< \beta$ and any $x\in%
\mathbb{R} $;

\item for every $x\in\mathbb{R} $ we have $\lim_{\alpha\to
0^{\pm}}f_{\alpha}(x)=x$;

\item if we set $c(\alpha)=\sup_{x\in\mathbb{R} }\{f_{\alpha}(x)-x\}$ then $%
\lim_{\alpha\to 0^{\pm}}c(\alpha)=0$.
\end{enumerate}

As in (\ref{1}), we choose
\begin{equation*}
\varphi (f_{\alpha },X)=-\inf_{x\in \mathbb{R}}\{x-f_{\alpha }(x)\}=c(\alpha
),
\end{equation*}%
which represents the risk reduction we would benefit by selling $f_{\alpha
}(X)$ jointly to the acquisition of $X$ (independently from the payoff of $X$%
). Equivalently we may interpret $\varphi (f_{\alpha },X)$ as the maximal
benefit we would realize buying $f_{\alpha }$ and selling the underlying.
Notice that it can be easily checked from the properties of the family $%
\{f_{\alpha }\}$ that $c(\cdot )$ is (strictly) increasing in $\alpha $.
\newline
The parameter $\alpha $, which indexes the family, represents the degree of
approximation: the higher $\alpha $ is the higher payoff the derivative $%
f_{\alpha }(X)$ will grant. The identity, i.e. when $\alpha =0$, clearly
corresponds to the case in which the risk completely annihilates by buying
and selling $X$. On the other hand for $\alpha >0$ (resp. $\alpha <0$) we
are considering testing functions which approximate from above (resp. from
below) the identity: therefore the strategy $X-f_{\alpha }(X)<0$ (resp. $%
X-f_{\alpha }(X)>0$) will bring losses (resp. gains) which are controlled by
$c(\alpha )$.

If we write explicitly the function $\Pi$ defined in \eqref{Pi} we obtain
the following formulation
\begin{eqnarray*}
\Pi(r,X;\mathbb{Q}) &=&\inf_{\alpha\in\mathbb{R} }\left\{ E_{\mathbb{Q}}%
\left[ f_{\alpha}(X) \right] \mid c(\alpha)\geq r\right\} \\
&=&\inf_{\alpha\in\mathbb{R} }\left\{ E_{\mathbb{Q}}\left[ f_{\alpha}(X) %
\right] \mid \alpha\geq c^{-1}(r)\right\} \\
&=& E_{\mathbb{Q}}\left[ f_{c^{-1}(r)}(X)\right],
\end{eqnarray*}%
with $c^{-1}$ being the left inverse of $c$.

From Proposition \ref{Price&Risk} we know that $R$ defined by \eqref{R} is a
Value\& Risk map (i.e. $R$ satisfies (1-2-3Mon), (QCo) and (CLI)). Moreover,
from $\varphi (f_{\alpha },X)=c(\alpha ),$ we obtain
\begin{eqnarray}
R(p,X;\mathbb{Q}) &=&\sup \{s\mid E_{\mathbb{Q}}\left[ f_{c^{-1}(s)}(X)%
\right] \leq p\}  \notag \\
&=&\sup \{\varphi (f_{s},X)\mid E_{\mathbb{Q}}\left[ f_{s}(X)\right] \leq
p\},  \label{R:example1}
\end{eqnarray}%
which interpretation is the following: \emph{the intrinsic Risk of acquiring
$X$ at price $p$ if we assume $\mathbb{Q}$ as the correct pricing rule ,
corresponds to the maximal risk reduction we would obtain buying $X$ and
selling a derivative $f_{\alpha }(X)$ with a price at most equal to $p$}.
Indeed implementing the strategy which buys $X$ and sells $f_{\alpha }(X)$,
has initial zero cost and guarantees that potential losses are at most given
by $\varphi (f_{\alpha },X)=c(\alpha )$.

\begin{properties}
Given the above definitions we have the following additional properties.

\begin{enumerate}
\item[P1] If for some $\alpha\in \mathbb{R} $ we have $p=E_{\mathbb{Q}}\left[
f_{\alpha}(X)\right]$ then $R(p,X;\mathbb{Q})=\varphi(f_{\alpha},X)=c(%
\alpha) $.

\smallskip

The proof of this property follows directly from the representation of $R$
given in \eqref{R:example1} and the properties of the family $\{f_{\alpha}\}$%
.

\item[P2] $R(a,a;\mathbb{Q})=0$ for $a\in\mathbb{R} $ and for every $\mathbb{%
Q} \in\mathcal{P}(\Omega)$.

\smallskip

We notice that $f_s(a)\leq a$ if and only if $s\leq 0$ so that $R(a,a;%
\mathbb{Q} )=\sup\{c(s)\mid f_s(a)\leq a\}\leq 0$. Moreover for $s=0$ we
have $c(s)=0$ and hence the thesis.

\item[P3] If $p \lesseqqgtr E_{\mathbb{Q} }[X]$ then $R(p,X;\mathbb{Q} )
\lesseqqgtr 0$. \

\smallskip

To show this last property we recall that for $\alpha>0$ (resp. $\alpha<0$
or $\alpha=0$) we have $f_{\alpha}(x)>x$ (resp. $f_{\alpha}(x)<x$ or $%
f_{\alpha}(x)=x$) for any $x\in\mathbb{R} $ so that $E_{\mathbb{Q}
}[f_{\alpha}(X)]>E_{\mathbb{Q} }[X]$ (resp. $E_{\mathbb{Q}
}[f_{\alpha}(X)]<E_{\mathbb{Q} }[X]$ or $E_{\mathbb{Q} }[f_{\alpha}(X)]=E_{%
\mathbb{Q} }[X]$). Then
\begin{eqnarray*}
R(E_{\mathbb{Q} }[X],X;\mathbb{Q} ) & = & \sup\left\{c(s)\mid E_{\mathbb{Q}
}[f_s(X)]\leq E_{\mathbb{Q} }[X]\right\}=0,
\end{eqnarray*}
and similarly for $p>E_{\mathbb{Q} }[X]$ or $p<E_{\mathbb{Q} }[X]$.
\end{enumerate}
\end{properties}

\subsection{Control of potential losses of the underlying using call options.%
}

Consider the set of call options $K=\{f(x)=(x-k)^{+}\mid k\in \mathbb{R}\}$
for any possible strike and assume that the agent is confident about the
fact that the underlying $X$ is bounded from below (i.e. $\inf_{\omega \in
\Omega }X(\omega )>-\infty $). For simplicity let the codomain of $X$ be
some convex subset of $\mathbb{R}$ (i.e. there are no gaps in the codomain)
of the type $(a,+\infty )$ or $[a,+\infty )$ (in particular $\sup_{\omega
\in \Omega }X(\omega )=+\infty $). Buying a call option written on $X$ for $%
k<0$ and selling $X$, guarantees both a positive performance and no expected
losses but the agent will have to pay a high price. If $k>0$ (for a lower
price) the agent will face controlled losses (for $X\geq 0$) and gains up to
a value $|a|$ for $X<0$. We consider
\begin{equation*}
\varphi ((x-k)^{+},X)=\inf_{\omega \in \Omega }\{(X(\omega )-k)^{+}-X(\omega
)\}=-k
\end{equation*}%
which represents the (pessimistic) payoff we will face in the worst case
scenario acquiring the derivative and selling the underlying.

We have immediately that
\begin{equation}
\Pi (r,X;\mathbb{Q})=\inf_{k\in \mathbb{R}}\left\{ E_{\mathbb{Q}}\left[
(X-k)^{+}\right] \mid \varphi ((x-k)^{+},X)\geq r\right\} =\int (X+r)^{+}d%
\mathbb{Q}  \label{pi:call}
\end{equation}%
To prove the previous equality we only notice that for $k_{n}\uparrow -r$ as
$n\rightarrow \infty $, we have $\int (X-k_{n})^{+}d\mathbb{Q}\downarrow
\int (X+r)^{+}d\mathbb{Q}$.

\bigskip

\textbf{Properties} From Proposition \ref{properties} we know that $R$
defined by \eqref{R} is (1-2-3Mon), (QCo) and (CLI) and hence is a Value and
Risk map. Moreover

\begin{description}
\item[P1] $R(E_{\mathbb{Q}}[(X-k)^{+}],X;\mathbb{Q})=\varphi ((x-k)^{+},X)=-k
$.

This property follows immediately from the definition. It states that if we
pay for $X$ the same price as a call options with $k<0$ (resp. $k>0$) then
we are facing an intrinsic risk given by $-k>0$ (resp. $-k<0$). This
intrinsic risk corresponds to the difference between the payoff of $(X-k)^{+}
$ and $X$ in the worst case scenario.

\item[P2] Cash Additivity: $R(p,X+a;\mathbb{Q})=R(p,X;\mathbb{Q})-a$.

To show this property we simply observe that
\begin{equation*}
\Pi (r,X+a;\mathbb{Q})=E_{\mathbb{Q}}[(X+(r+a))^{+}]=\Pi (r+a,X;\mathbb{Q})
\end{equation*}%
and therefore $R(p,X+a;\mathbb{Q})=\sup \{s\mid \Pi (s+a,X;\mathbb{Q})\leq
p\}=\sup \{s-a\mid \Pi (s,X;\mathbb{Q})\leq p\}$, which gives the thesis.

\item[P3] Let $X\geq 0$. If $p=E_{\mathbb{Q}}[X]$ (resp. $p\gtrless E_{%
\mathbb{Q}}[X]$) then $R(p,X;\mathbb{Q})=0$ (resp. $R(p,X;\mathbb{Q}%
)\gtrless 0$ ). In particular $R(a,a;\mathbb{Q})=0$ for $a\in \lbrack
0,+\infty )$ and for every $\mathbb{Q}\in \mathcal{P}(\Omega )$.

The property follows from $\sup \left\{ s\in \mathbb{R}\mid E_{\mathbb{Q}%
}[(X+s)^{+}]\leq E_{\mathbb{Q}}[X]\right\} =0$. Indeed for any $s>0$ $E_{%
\mathbb{Q}}[(X+s)^{+}]=E_{\mathbb{Q}}[X]+s>E_{\mathbb{Q}}[X]$ and since $%
X\geq 0$ we have $E_{\mathbb{Q}}[(X)^{+}]=E_{\mathbb{Q}}[X]$.

\item[P4] Let $E_{\mathbb{Q}}[X]<0$ then $R(E_{\mathbb{Q}}[X],X;\mathbb{Q}%
)=-\infty $. In particular $R(a,a;\mathbb{Q})=-\infty $ for $a\in (-\infty
,0)$ and for every $\mathbb{Q}\in \mathcal{P}(\Omega )$.

The proof is straightforward. Clearly the set of call options is not
feasible for testing a position $X$ which has a strong negative component
from the $\mathbb{Q}$-perspective.
\end{description}

\subsection{Asymmetric Tail Control with concave derivatives}

\label{esempio:tail} In this example we consider the case in which the
underlying $X\in \mathcal{L} (\Omega ,\mathcal{F})$ produces a potentially
unbounded loss together with potentially unbounded gains, in particular $%
\sup_{\omega\in\Omega} X^+(\omega)=\sup_{\omega\in\Omega}
X^-(\omega)=+\infty $. We here consider a fairly general class of concave
and increasing test functions described by a family $\{f_{\alpha}\}_{\alpha%
\in \mathbb{R}}$ such that

\begin{enumerate}
\item $\{f_{\alpha}\}_{\alpha\in \mathbb{R} }\subset C^0_+\cap C^1$, $%
f_{\alpha}$ concave for every $\alpha$ and $f_0(0)=0, f^{\prime}_{0}(0)=1$;

\item $f_{\alpha}(x)< f_{\beta}(x)$ for every $\alpha< \beta$ and any $x\in%
\mathbb{R} $;

%
\end{enumerate}

The higher $\alpha $ is the higher payoff the derivative $f_{\alpha }(X)$
will grant. Moreover for $\alpha <0$ the derivative $f_{\alpha }(X)$ will be
strictly dominated by $X$ in any possible state of nature.

The agent sells $f_{\alpha }(X)$ and keeps $X$: for $\alpha >0$ she will be
guaranteed a positive payoff in any case (with a reduced performance on the
positive tail). For $\alpha <0$ the agent may suffer a controlled loss on
the set $\{\omega \in \Omega \mid a_{1}<X(\omega )<a_{2}\}$ with $%
a_{1},a_{2}\in \mathbb{R}\cup \{\infty \}$ being the intersections of $%
f_{\alpha }(x)$ with the identity function $id(x)=x$. The residual risk in
the worst case scenario left over after selling $f_{\alpha }(X)$ jointly to
the acquisition of $X$ (independently from the payoff of $X$) is given by
\begin{equation*}
\varphi (f_{\alpha },X)=-\inf_{x\in \mathbb{R}}\{x-f_{\alpha }(x)\}=c(\alpha
).
\end{equation*}%
Easy computations show that $c(\alpha )=\bar{x}-f_{\alpha }(\bar{x})$ with $%
\bar{x}$ solution of $f_{\alpha }^{\prime }(x)=1$. \newline
Notice that from the concavity of $\{f_{\alpha }\}_{\alpha \in \mathbb{R}}$
and Theorem \ref{Price&Risk} we know that $R$ defined by \eqref{R} is a
Value\& Risk map (i.e. $R$ is (1-2-3Mon), (QCo) and (CLI)) satisfying in
addition the property (QCoX).

To obtain more explicit results we specify $f_{\alpha }(x)=g(\alpha )-e^{-x}$
with $g(0)=1$ and $g$ strictly increasing. In this case $c(\alpha )=g(\alpha
)-1$ and we can write explicitly the function $\Pi $ defined in \eqref{Pi}
\begin{eqnarray*}
\Pi (r,X;\mathbb{Q}) &=&\inf_{\alpha \in \mathbb{R}}\left\{ E_{\mathbb{Q}}%
\left[ f_{\alpha }(X)\right] \mid c(\alpha )\geq r\right\} \\
&=&E_{\mathbb{Q}}\left[ f_{g^{-1}(r+1)}(X)\right] =r+1-E_{\mathbb{Q}%
}[e^{-X}].
\end{eqnarray*}%
Moreover we can observe that
\begin{eqnarray}
R(p,X;\mathbb{Q}) &=&\sup \{s\mid E_{\mathbb{Q}}\left[ f_{g^{-1}(r+1)}(X)%
\right] \leq p\}  \notag \\
&=&\sup \{\varphi (f_{s},X)\mid E_{\mathbb{Q}}\left[ f_{s}(X)\right] \leq
p\}=p-E_{\mathbb{Q}}[1-e^{-X}],  \label{R:example2}
\end{eqnarray}%
which interpretation is the following: the intrinsic Risk of acquiring $X$
at price $p$ is the discrepancy between the price $p$ and the value of the
derivative in $0$ under the probability measure $\mathbb{Q}$.

\subsection{Insured testing functions}

Given a position $X\in \mathcal{L}(\Omega ,\mathcal{F})$ we face the problem
that the risk connected to $X$ might be unbounded but the agent is not
willing to sell $X$. For this reason she aims at buying an insurance on $X$
in order to control the risk. The set of insurances on $X$ is denoted by $%
\mathcal{I}\subseteq C^{0}$ and we assume that the price of each insurance $%
i\in \mathcal{I}$ is exogenously determined by a (possibly non-liner)
functional $c(\cdot )$. The only requirements on $c$ are positivity ($%
c(i)\geq 0$ for every $i\geq 0$) and monotonicity. Let $\mathbb{Q}\in
\mathcal{P}$ be a pricing rule, then the agent will try to minimize the cost
under $\mathbb{Q}$ under the risk constraint for the aggregated position. In
this case
\begin{equation*}
\mathcal{K}=\{f(x)=x+i(x)-c(i)\mid i\in \mathcal{I}\}.
\end{equation*}%
For this reason the agent will face the minimization problem

\begin{eqnarray*}
&& \Pi(r,X;\mathbb{Q}) = \inf_{f\in K}\left\{ E_{\mathbb{Q}}\left[ f(X)%
\right] \mid \varphi (f(X))\geq r\right\} \\
& =& E_{\mathbb{Q}}\left[X\right]+\inf_{i\in \mathcal{I}}\left\{ E_{\mathbb{Q%
}}\left[i(X)\right]-c(i) \mid \varphi (X+i(X)-c(i))\geq r\right\}
\end{eqnarray*}

Here $\Pi (r,X;\mathbb{Q})$ represents the price under $\mathbb{Q}$ of the
cheapest insured position which guarantees that the risk diminishes at least
by $r$. If for some $r$ this problem is solved by $i^{\ast }$ and if $p$ is
the observed price for $X$, then we have an equilibrium between observed
prices and pricing beliefs under $\mathbb{Q}$ i.e.
\begin{equation*}
p+c(i^{\ast })=E_{\mathbb{Q}}[X+i^{\ast }(X)]
\end{equation*}%
with a risk residual equal to $r$.

\paragraph{Insuring using put options.}

Consider the special case with $p$ representing the initial value of an
unbounded from below risky position $X\in \mathcal{L}(\Omega ,\mathcal{F})$
(i.e. $\inf_{\omega \in \Omega }X(\omega )=-\infty $). We insure $X$ using
put options, so that
\begin{equation*}
K=\{f(x)=x+(k-x)^{+}-c(k)\mid k\in \mathbb{R}\}
\end{equation*}%
and it is easy to show that $K\subset C_{+}^{0}$. Here $c:\mathbb{R}%
\rightarrow \lbrack 0,+\infty )$ is assumed to be strictly increasing and $%
k-c(k)$ is strictly increasing with $\lim_{k\rightarrow +\infty }(k-c(k))>0$%
. Moreover we are interested at the case where $\inf_{\omega }X(\omega
)=-\infty $ and choose consequently $\varphi (f,X)=\inf_{\omega }f(X(\omega
))=k-c(k)$ which is independent from $X$ and represents the worst case
payoff of the insured position $f(X)$. We compute
\begin{eqnarray*}
\Pi (r,X;\mathbb{Q}) &=&E_{\mathbb{Q}}\left[ X\right] +\inf_{k\in \mathbb{R}%
}\left\{ E_{\mathbb{Q}}\left[ (k-X)^{+}-c(k)\right] \mid k-c(k)\geq r\right\}
\\
&=&E_{\mathbb{Q}}\left[ X\right] +E_{\mathbb{Q}}\left[ (b(r)-X)^{+}\right]
-c(b(r)),
\end{eqnarray*}%
where $b$ is the inverse of the function $y-c(y)$. Let $R$ be the map
defined by \eqref{R}. Simple inspections together with Proposition \ref%
{properties} show that $R$ is (1-2-3Mon), (QCo) and (CLI) and hence is a
Value and Risk map. Indeed if $p=E_{\mathbb{Q}}[X]$ we have
\begin{equation*}
R(E_{\mathbb{Q}}[X],X;\mathbb{Q})=\max \{s\mid E_{\mathbb{Q}}\left[
(b(s)-X)^{+}\right] =c(b(s))\}
\end{equation*}%
The proof follows from the definition. The interpretation is that the risk
we face buying $X$ at price $E_{\mathbb{Q}}[X]$ corresponds to the index $s$
of the insuring derivative for which the exogenous price $c(s)$ corresponds
to the price $E_{\mathbb{Q}}\left[ (b(s)-X)^{+}\right] $ computed using $%
\mathbb{Q}$.

%
%
%
%
%

\section{Appendix A\label{AppA}}

\label{GeneralResults}

Let $\mathbb{Q}^{1},\mathbb{Q}^{2}\in \mathcal{P}(\Omega )$, $X\in \mathcal{L%
}(\Omega ,\mathcal{F})$ and let $C_{+}^{0}:=\{f\in C^{0}\mid f\text{
increasing }\}.$ Recall the notation $\mathbb{Q}^{1}\preccurlyeq _{X}\mathbb{%
Q}^{2}$ given in Definition \ref{defOrder} and the fact that $\mathbb{Q}%
^{1}\preccurlyeq _{1}\mathbb{Q}^{2}$ implies $\int fdQ^{1}\leq \int fdQ^{2}$
for any $f\in C_{+}^{0}$.

The following Proposition collects the properties of $\Pi $ and $R$ and can
be used to check that Theorem \ref{Price&Risk} holds true.

\begin{proposition}
\label{properties}[Monotonicity and convexity] Consider $\Pi $, $R$ defined
respectively by (\ref{Pi}) and (\ref{R}). Then we have the following
properties.

\begin{itemize}
\item[(a1)] $\Pi (\cdot ,X;\mathbb{Q})$ is monotone increasing and $R(\cdot
,X;\mathbb{Q})$ is (1Mon).

\item[(a2)] $\Pi (r,X;\cdot )$ is concave and $R(p,X;\cdot )$ is quasiconvex;

\item[(b)] If $K\subseteq C_{+}^{0}(\mathbb{R})$ then:

\begin{itemize}
\item[(b1)] $\mathbb{Q}^{1}\preccurlyeq _{X}\mathbb{Q}^{2}$ implies $\Pi
(r,X;\mathbb{Q}^{1})\leq \Pi (r,X;\mathbb{Q}^{2})$ and $R$ is (3Mon);

\item[(b2)] $\mathbb{Q}^{1}\preccurlyeq _{X}\mathbb{Q}^{2}$ and $p_{1}\geq
p_{2}$ implies $R(p_{1},X;\mathbb{Q}^{1})\geq R(p_{2},X;\mathbb{Q}^{2})$;
\end{itemize}

\item[(c)] Suppose that $\varphi $ satisfies
\begin{equation}
X\leq Y\Longrightarrow \varphi (f,X)\geq \varphi (f,Y)\text{ for all }k\in K.
\label{monphi}
\end{equation}%
Then $X\leq Y$ implies $\Pi (r,X;\mathbb{Q})\geq \Pi (r,Y;\mathbb{Q})$ and $%
R $ is (2Mon).

\item[(d)] If $K\subseteq \{f\in C_{+}^{0}(\mathbb{R})\mid f\text{ concave }%
\}$ and $\varphi $ satisfies (\ref{equal1}) then $\Pi (r,\cdot ;\mathbb{Q})$
is concave and $R(p,\cdot ,\mathbb{Q})$ is quasiconvex.

\item[(e)] (CLI) Suppose that
\begin{equation}
\{f\in K\mid \varphi (f,Y)\geq r\}=\{f\in K\mid \varphi (f,X)\geq r\}\text{
for all }f\in K\text{.}  \label{fK}
\end{equation}%
For any $(X,\mathbb{Q}^{1}),(Y,\mathbb{Q}^{2})\in \mathcal{L}(\Omega .%
\mathcal{F})\times \mathcal{P}(\Omega )$ such that $Q_{X}^{1}=Q_{Y}^{2}$
then $\Pi (r,X;\mathbb{Q}^{1})=\Pi (r,Y;\mathbb{Q}^{2})$ and $R(p,X;\mathbb{Q%
}^{1})=R(p,Y;\mathbb{Q}^{2})$, so that $R$ satisfies (CLI).
\end{itemize}
\end{proposition}

Notice that if $\varphi $ satisfies (\ref{equal1}) then (\ref{monphi}) and (%
\ref{fK}) hold true.

\bigskip

\begin{proof}
\noindent (a1) follows from the definition.

\noindent (a2) The concavity of $\Pi (r,X;\cdot )$ follows from its
definition and the properties of the $\inf $. Take $\mathbb{Q}^{1},\mathbb{Q}%
^{2}$ and $\lambda \in (0,1)$ and let $R(p,X;\mathbb{Q}^{1}):=t_{1}\leq
t_{2}:=R(p,X;\mathbb{Q}^{2})$. In this proof we omit the dependence on $X$.
We need to prove that $R(p,\lambda \mathbb{Q}^{1}+(1-\lambda )\mathbb{Q}%
^{2})\leq t_{2}$. Note that $t_{i}=\sup \{s\mid \Pi (s,\mathbb{Q}^{i})\leq
p\}$. Then: $\Pi (s,\mathbb{Q}^{i})>p$ for all $s>t_{i}$. Therefore, $\Pi (s,%
\mathbb{Q}^{1})>p$ and $\Pi (s,\mathbb{Q}^{2})>p$ for all $s>t_{2}$ implies:
$\lambda \Pi (s,\mathbb{Q}^{1})+(1-\lambda )\Pi (s,\mathbb{Q}^{2})>p$ for
all $s>t_{2}$. As a consequence:%
\begin{equation*}
\sup \left\{ s\mid \lambda \Pi (s,\mathbb{Q}^{1})+(1-\lambda )\Pi (s,\mathbb{%
Q}^{2})\leq p\right\} \leq t_{2}.
\end{equation*}%
From the concavity of $\Pi (s,\cdot ),$ $\lambda \Pi (s,\mathbb{Q}%
^{1})+(1-\lambda )\Pi (s,\mathbb{Q}^{2})\leq \Pi (s,\lambda \mathbb{Q}%
^{1}+(1-\lambda )\mathbb{Q}^{2})$, we obtain:%
\begin{equation*}
\left\{ s\mid \Pi (s,\lambda \mathbb{Q}^{1}+(1-\lambda )\mathbb{Q}^{2})\leq
p\right\} \subseteq \left\{ s\mid \lambda \Pi (s,\mathbb{Q}^{1})+(1-\lambda
)\Pi (s,\mathbb{Q}^{2})\leq p\right\} .
\end{equation*}%
Hence:%
\begin{eqnarray*}
R(p,\lambda \mathbb{Q}^{1}+(1-\lambda )\mathbb{Q}^{2}) &=&\sup \left\{ s\mid
\Pi (s,\lambda \mathbb{Q}^{1}+(1-\lambda )\mathbb{Q}^{2})\leq p\right\} \\
&\leq &\sup \left\{ s\mid \lambda \Pi (s,\mathbb{Q}^{1})+(1-\lambda )\Pi (s,%
\mathbb{Q}^{2})\leq p\right\} \leq t_{2}.
\end{eqnarray*}

\noindent (b1) Recall $\mathbb{Q}^{1}\preccurlyeq _{X}\mathbb{Q}^{2}$
implies $\int fdQ_{X}^{1}\leq \int fdQ_{X}^{2}$ for any $f\in C_{+}^{0}$.
Clearly this shows that $\mathbb{Q}^{1}\preccurlyeq _{X}\mathbb{Q}^{2}$
implies $\Pi (s,X;\mathbb{Q}^{1})\leq \Pi (s,X;\mathbb{Q}^{2})$. As a
consequence $\{s\mid \Pi (s,X;\mathbb{Q}^{1})\leq p\}\supseteq \{s\mid \Pi
(s,X;\mathbb{Q}^{2})\leq p\}$ and $R(p,X;\mathbb{Q}^{1})\geq R(p,X;\mathbb{Q}%
^{2})$.

\noindent (b2) Follows from (a1) and (b1).

\noindent (c) Let $X\leq Y$. Then (\ref{monphi}) implies $\{f\in K\mid
\varphi (f,Y)\geq r\}\subseteq \{f\in K\mid \varphi (f,X)\geq r\}$ so that
\begin{eqnarray*}
\Pi (r,Y;\mathbb{Q}) &=&\inf \left\{ E_{Q}[f(Y)]\mid \varphi (f,Y)\geq
r\right\} \geq \inf \left\{ E_{Q}[f(Y)]\mid \varphi (f,X)\geq r\right\} \\
&\geq &\inf \left\{ E_{Q}[f(X)]\mid \varphi (f,X)\geq r\right\} =\Pi (r,X;%
\mathbb{Q})
\end{eqnarray*}%
Moreover from this property we have $\{s\mid \Pi (s,Y;\mathbb{Q})\leq
p\}\subseteq \{s\mid \Pi (s,X;\mathbb{Q})\leq p\}$ so that $R(p,Y;\mathbb{Q}%
)\leq R(p,X;\mathbb{Q})$.

\noindent (d) The concavity of $\Pi (r,\cdot ;\mathbb{Q})$ follows from $%
K\subseteq \{f\in C_{+}^{0}(\mathbb{R})\mid f\text{ concave }\}$, the
properties of the $\inf $ and (\ref{equal1}). Take $X_{1},X_{2}$ and $%
\lambda \in (0,1)$ and let $R(p,X;\mathbb{Q}):=t_{1}\leq t_{2}:=R(p,X;%
\mathbb{Q})$. In this proof we omit the dependence on $\mathbb{Q}$. We need
to prove that $R(p,\lambda X_{1}+(1-\lambda )X_{2})\leq t_{2}$. As before $%
\Pi (s,X_{i})>p$ for all $s>t_{i}$. Therefore, $\Pi (s,X_{1})>p$ and $\Pi
(s,X_{2})>p$ for all $s>t_{2}$. For any $\lambda \in (0,1)$ we have $\lambda
\Pi (s,X_{1})+(1-\lambda )\Pi (s,X_{2})>p$ for all $s>t_{2}$. This implies:%
\begin{equation*}
\sup \left\{ s\mid \lambda \Pi (s,X_{1})+(1-\lambda )\Pi (s,X_{2})\leq
p\right\} \leq t_{2}.
\end{equation*}%
From the concavity of $\Pi (s,\cdot ),$ we obtain:
\begin{equation*}
\left\{ s\mid \Pi (s,\lambda X_{1}+(1-\lambda )X_{2})\leq p\right\}
\subseteq \left\{ s\mid \lambda \Pi (s,X_{1})+(1-\lambda )\Pi (-s,X_{2})\leq
p\right\} .
\end{equation*}%
Hence:%
\begin{eqnarray*}
R(p,\lambda X_{1}+(1-\lambda )X_{2}) &=&\sup \left\{ s\mid \Pi (s,\lambda
X_{1}+(1-\lambda )X_{2})\leq p\right\} \\
&\leq &\sup \left\{ s\mid \lambda \Pi (s,X_{1})+(1-\lambda )\Pi
(s,X_{2})\leq p\right\} \leq t_{2}.
\end{eqnarray*}%
(e) Follows directly from the definitions and (\ref{fK}).
\end{proof}

\bigskip

\begin{proof}[Proof of Proposition \ref{dependence:K}]
The first property is straightforward. To prove (\ref{PIqco})
observe that
\begin{eqnarray*}
&&\lambda \Pi _{K^{1}}(r)+(1-\lambda )\Pi _{K^{2}}(r)
\\&=&\inf_{f^{1}\in K^{1}}\left\{ E_{\mathbb{Q}}\left[ \lambda
f^{1}(X)\right] \mid \varphi
(f^{1},X)\geq r\right\}  \\
&+&\inf_{f^{2}\in K^{2}}\left\{ E_{\mathbb{Q}}\left[ (1-\lambda )f^{2}(X)%
\right] \mid \varphi (f^{2},X)\geq r\right\}  \\
&=&\inf_{f^{1}\in K^{1},\text{ }f^{2}\in K^{2}}\left\{
E_{\mathbb{Q}}\left[
\lambda f^{1}(X)+(1-\lambda )f^{2}(X)\right] \mid \varphi (f^{1},X)\geq r%
\text{ and }\varphi (f^{2},X)\geq r\right\}  \\
&\geq &\inf_{f^{1}\in K^{1},\text{ }f^{2}\in K^{2}}\left\{ E_{\mathbb{Q}}%
\left[ f^{\lambda }(X)\right] \mid \varphi (f^{\lambda },X)\geq r\right\}  \\
&=&\inf_{f\in K^{\lambda }}\left\{ E_{\mathbb{Q}}\left[
f(X)\right] \mid \varphi (f,X)\geq r\right\} =\Pi _{K^{\lambda
}}(r).
\end{eqnarray*}%
The previous inequality is motivated by the quasi-concavity of
$\varphi ,$
namely%
\begin{equation*}
\varphi (f^{1},X)\geq r\text{ and }\varphi (f^{2},X)\geq r\text{ implies }%
\varphi (f^{\lambda },X)\geq r
\end{equation*}%
Now we show (\ref{Rqco}).
\begin{eqnarray*}
R_{K^{\lambda }}(p) &=&\sup \left\{ s\in \mathbb{R}\mid \Pi
_{K^{\lambda
}}(s)\leq p\right\}  \\
&\geq &\sup \left\{ s\in \mathbb{R}\mid \lambda \Pi
_{K^{1}}(s)+(1-\lambda
)\Pi _{K^{2}}(s)\leq p\right\}  \\
&\geq &\sup \left\{ s\in \mathbb{R}\mid \Pi _{K^{1}}(s)\leq
p\text{ and }\Pi
_{K^{2}}(s)\leq p\right\}  \\
&=&R_{K^{1}}(p)\wedge R_{K^{2}}(p).
\end{eqnarray*}%
where the first inequality follows from (\ref{PIqco}).

The first inequality in Item 3 follows exactly with the same
argument used in Item 2. The second inequality of Item 3 then is a
consequence of:
\begin{eqnarray*}
R_{K^{1}+K^{2}}(p) &\geq &\sup \left\{ s\in \mathbb{R}\mid \Pi
_{K^{1}}(s)+\Pi _{K^{2}}(s)\leq p\right\}  \\
&\geq &\sup \left\{ s\in \mathbb{R}\mid \Pi _{K^{1}}(s)\leq \alpha
p\text{
and }\Pi _{K^{2}}(s)\leq (1-\alpha )p\right\}  \\
&= &R_{K^{1}}(\alpha p)\wedge R_{K^{2}}((1-\alpha )p),\text{ for any }%
\alpha \in \mathbb{R}
\end{eqnarray*}
\end{proof}

The proof of the following Proposition is omitted since is a straightforward
consequence of the definitions.

\begin{proposition}[Behavior with respect to cash]\label{cash}
Consider $\Pi $, $R $ defined respectively by (\ref{Pi}) and (\ref{R}). Let $%
r,p,\alpha \in \mathbb{R} $ then we have the following properties.

\begin{enumerate}
\item[(a)] $R(p,p;\mathbb{Q} )=0$ if $\Pi(r,p;\mathbb{Q} )=p+r$.

\item[(b)] $R(p+\alpha,X;\mathbb{Q})=\alpha+R(p,X;\mathbb{Q})$ if $%
\Pi(r+\alpha,X;\mathbb{Q})=\alpha+\Pi(r,X;\mathbb{Q})$.

\item[(c)] $R(p,X+\alpha;\mathbb{Q} )=R(p,X;\mathbb{Q})-\alpha$ if $%
\Pi(r+\alpha,X;\mathbb{Q})=\Pi(r,X+\alpha;\mathbb{Q} )$.

\item[(d)] $R(p+\alpha,X+\alpha;\mathbb{Q} )=R(p,X;\mathbb{Q})$ if $%
\Pi(r,X+\alpha;\mathbb{Q})=\Pi(r,X;\mathbb{Q} )+\alpha$.
\end{enumerate}

%
%
\end{proposition}


In the following Proposition we provide some explicit forms for $\Pi $, when
we are able to find a representative class of one parameter functions.

\begin{proposition}
\label{trasf} For any fixed $X\in\mathcal{L}(\Omega,\mathcal{F} )$ consider $%
\varphi:K\times \mathcal{L}(\Omega,\mathcal{F} )\rightarrow \mathbb{R} \cup {%
+\infty}$ and let $I_{\varphi}=\inf_{f\in K}\varphi(f,X)$. Assume that there
exist $f^0\in K$ and a one parameter class of transformations $%
\{T_{\alpha}\}_{\alpha\geq I_{\varphi}}$ such that

\begin{itemize}
\item $T_{\alpha}:K\rightarrow K$ with $T_0f^0=f^0 $and $\varphi(T_{%
\alpha}f^0,X)=\alpha$,

\item $T_{\alpha} f^0 < T_{\beta}f^0$ for $\alpha< \beta$,

\item for any $g\in K$ such that $\varphi(g,X)=\alpha$ we have $T_{\alpha}
f^0\leq g$ for $\alpha\leq 0$ (resp. $T_{\alpha} f^0\geq f^0$ $\alpha\geq 0$%
).
\end{itemize}

Then

\begin{itemize}
\item[(a)] $\Pi(r,X;\mathbb{Q}) =\inf_{f\in K}\left\{ E_{\mathbb{Q}}\left[
f(X)\right] \mid \varphi (f,X)= r\right\}$.

\item[(b)] $\Pi(r,X;\mathbb{Q}) =\inf_{\alpha\geq I_{\varphi}} \left\{ E_{%
\mathbb{Q}}\left[ T_{\alpha}\circ f^0 (X) \right] \mid \varphi
(T_{\alpha}\circ f^0,X)\geq r\right\}$;

\item[(c)] $\Pi(r,X;\mathbb{Q}) = \; E_{\mathbb{Q}}\left[ T_{r} f^0(X)\right]
$;

\item[(d)] $R(E_{\mathbb{Q} }[T_{\alpha}f^0(X)],X;\mathbb{Q}%
)=\varphi(T_{\alpha} f^0(X))=\alpha.$
\end{itemize}
\end{proposition}

\begin{proof}
(1) Obviously $\Pi(r,X;\mathbb{Q}) \leq \inf_{f\in K}\left\{ E_{\mathbb{Q}}%
\left[ f(X)\right] \mid \varphi (f,X)= r\right\}$. By contradiction assume $%
< $. Then there exists $\overline{f}\in K$ such that $\varphi(\overline{f},
X)>r$ and
\begin{equation*}
E_{\mathbb{Q} }[\overline{f}(X)]<\inf_{f\in K}\left\{ E_{\mathbb{Q}}\left[
f(X)\right] \mid \varphi (f, X)= r\right\}.
\end{equation*}
Take $\alpha=\varphi(\overline{f}, X)$: in this way we have $%
\varphi(T_{\alpha}\circ f_0, X)= \varphi(\overline{f}, X)$ which implies $%
T_{\alpha} f^0\leq \overline{f}$. Moreover if $\varepsilon=\alpha-r>0$ we
have $T_{\alpha-\varepsilon} f^0\leq T_{\alpha} f^0\leq \overline{f}$ with $%
\varphi(T_{\alpha-\varepsilon} f^0,X)=r$ and thus a contradiction since $E_{%
\mathbb{Q} }[T_{\alpha-\varepsilon}\circ f_0,X]\leq E_{\mathbb{Q} }[%
\overline{f}(X)]< \inf_{f\in K}\left\{ E_{\mathbb{Q}}\left[ f(X)\right] \mid
\varphi (f,X)= r\right\}$.

(b) and (c) From the previous step
\begin{equation*}
\Pi (r,X;\mathbb{Q})=\inf_{f\in K}\left\{ E_{\mathbb{Q}}\left[ f(X)\right]
\mid \varphi (f,X)=r\right\} \leq E_{\mathbb{Q}}\left[ T_{r}f^{0}(X)\right]
\end{equation*}%
By assumption we also have $T_{r}f^{0}\leq g$ for any $g\in K$ such that $%
\varphi (g,X)=r$ and hence $E_{\mathbb{Q}}\left[ T_{r}f^{0}(X)\right] \leq
\Pi (r,X;\mathbb{Q})$. Thus $T_{r}f^{0}$ corresponds to the minimizer.
\newline
(4) Follows directly from the definition of $R(E_{\mathbb{Q}}[T_{\alpha
}f^{0}(X)],X;\mathbb{Q})$.
\end{proof}

\subsection{Duality for testing functions cones}

This section is devoted to the proof of Theorem \ref{thA}, which is
instrumental to Proposition \ref{dualityK}. In this section we will often
omit in the notations the dependence from $\mathbb{R}$ and write $ca$ for $%
ca(\mathbb{R})$, similarly for the other symbols.

\bigskip

Let $\mathcal{K}\subseteq \mathcal{L}^{\infty }(\Omega ,\mathcal{F})$ and $%
\varphi :\mathcal{K}\rightarrow \overline{\mathbb{R}}$. Let $\Pi _{\varphi }:%
\mathbb{R}\times ca\rightarrow \overline{\mathbb{R}}$ be defined by:
\begin{equation}
\Pi _{\varphi }(r,\mu ):=\inf_{Y\in \mathcal{K}}\left\{ E_{\mu }[Y]\mid
\varphi (Y)\geq r\right\} ,  \label{Q}
\end{equation}%
and let $R_{\varphi }:\mathbb{R}\times ca\rightarrow \overline{\mathbb{R}}$
be the right inverse of the increasing function $\Pi (\cdot ,\mu )$
\begin{equation}
R_{\varphi }(p,\mu ):=\sup \left\{ r\in \mathbb{R}\mid \Pi _{\varphi }(r,\mu
)\leq p\right\} .  \label{H2}
\end{equation}%
Let $H:\mathbb{R}\times ca\rightarrow \overline{\mathbb{R}}$ be defined by
\begin{equation}
H_{\varphi }(p,\mu ):=\sup_{\xi \in \mathcal{K}}\left\{ \varphi (\xi )\mid
E_{\mu }[\xi ]\leq p\right\} .  \label{HH}
\end{equation}%
Notice that the three functions $\Pi _{\varphi }(\cdot ,\mu ),$ $R_{\varphi
}(\cdot ,\mu )$\ and $H_{\varphi }(\cdot ,\mu )$ are monotone increasing. In
the proofs we will omit the label $\varphi $ in $\Pi _{\varphi },$ $%
R_{\varphi }$ and $H_{\varphi }.$

\begin{proposition}
\label{propHH}Let $H_{\varphi }^{+}(\cdot ,\mu )$ be the right continuous
version of $H_{\varphi }(\cdot ,\mu )$. Then:%
\begin{equation}
H_{\varphi }^{+}(p,\mu ):=\inf_{s>p}H_{\varphi }(s,\mu )=R_{\varphi }(p,\mu
).  \label{HRinf}
\end{equation}
\end{proposition}

\begin{proof}
Since $R(\cdot ,\mu )$ is the right inverse of the increasing function $\Pi
(\cdot ,\mu ),$ $R(\cdot ,\mu )$ is right continuous. To prove that $%
H^{+}(p,\mu )\leq R(p,\mu )$ it is sufficient to show that for all $t>p$ we
have:%
\begin{equation}
H(t,\mu )\leq R(t,\mu ),  \label{50}
\end{equation}%
Indeed, if (\ref{50}) is true%
\begin{equation*}
H^{+}(p,\mu )=\inf_{t>p}H(t,\mu )\leq \inf_{t>p}R(t,\mu )=R(p,\mu ),
\end{equation*}%
as both $H^{+}$ and $R$ are right continuous in the first argument.\newline
Writing explicitly the inequality (\ref{50})
\begin{equation*}
\sup_{\xi \in \mathcal{K}}\left\{ \varphi (\xi )\mid E_{\mu }[\xi ]\leq
t\right\} \leq \sup \left\{ r\in \mathbb{R}\mid \Pi (r,\mu )\leq t\right\}
\end{equation*}%
and letting $\xi \in \mathcal{K}$ satisfying $E_{\mu }[\xi ]\leq t$, we see
that it is sufficient to show the existence of $r\in \mathbb{R}$ such that $%
\Pi (r,\mu )\leq t$ and $r\geq \varphi (\xi )$. If $\varphi (\xi )=\infty $
then $\Pi (r,\mu )\leq t$ for any $r$ and therefore $R(t,\mu )=H(t,\mu
)=\infty $.

Suppose now that $\infty >\varphi (\xi )>-\infty $ and define $r:=\varphi
(\xi ).$ As $E_{\mu }[\xi ]\leq t$ we have:%
\begin{equation*}
\Pi (r,\mu ):=\inf \left\{ E_{\mu }[\xi ]\mid \varphi (\xi )\geq r\right\}
\leq t.
\end{equation*}%
Then $r\in \mathbb{R}$ satisfies the required conditions.

To obtain $H^{+}(p,\mu ):=\inf_{t>p}H(t,\mu )\geq R(p,\mu )$ it is
sufficient to prove that, for all $t>p,$ $H(t,\mu )\geq R(p,\mu )$, that is :%
\begin{equation}
\sup_{\xi \in \mathcal{K}}\left\{ \varphi (\xi )\mid E_{\mu }[\xi ]\leq
t\right\} \geq \sup \left\{ r\in \mathbb{R}\mid \Pi (r,\mu )\leq p\right\} .
\label{52}
\end{equation}%
Fix any $t>p$\ and consider any $r\in \mathbb{R}$ such that $\Pi (r,\mu
)\leq p$. By the definition of $\Pi $, for all $\varepsilon >0$ there exists
$\xi _{\varepsilon }\in \mathcal{K}$ such that $\varphi (\xi _{\varepsilon
})\geq r$ and $E_{\mu }[\xi _{\varepsilon }]\leq p+\varepsilon .$ Take $%
\varepsilon $ such that $0<\varepsilon <t-p$. Then $E_{\mu }[\xi
_{\varepsilon }]\leq t$ and $\varphi (\xi _{\varepsilon })\geq r$ and (\ref%
{52}) follows.
\end{proof}

\bigskip

From now on we suppose that $\mathcal{K}\subseteq \mathcal{L}^{\infty
}(\Omega ,\mathcal{F})$ is a convex cone. Consider $\mathcal{K}^{+}=\mathcal{%
K}\cap \mathcal{L}_{+}^{\infty }$ and define%
\begin{eqnarray*}
\mathcal{K}^{\circ } &=&\{\mu \in ca\mid \mu (\xi )\geq 0\;\forall \xi \in
\mathcal{K}^{+}\}. \\
\mathcal{K}_{1}^{\circ } &=&\{Q\in \mathcal{K}^{\circ }\mid Q(1_{\Omega
})=1\}.
\end{eqnarray*}

\begin{proposition}
\label{propHR}Let $\mathcal{M}\subseteq \mathcal{K}^{\circ }$ and $\Psi
_{\varphi }:\mathcal{K}\rightarrow \overline{\mathbb{R}}$ be defined by:
\begin{equation*}
\Psi _{\varphi }(Y):=\inf_{Q\in \mathcal{M}}H_{\varphi }(E_{Q}[Y],Q).
\end{equation*}%
If $\Psi _{\varphi }$ is continuous from above ($Y_{n}\downarrow Y$ implies $%
\Psi _{\varphi }(Y_{n})\downarrow \Psi _{\varphi }(Y)$), then:%
\begin{equation*}
\Psi _{\varphi }(Y):=\inf_{Q\in \mathcal{M}}H_{\varphi
}(E_{Q}[Y],Q)=\inf_{Q\in \mathcal{M}}H_{\varphi }^{+}(E_{Q}[Y],Q)=\inf_{Q\in
\mathcal{M}}R_{\varphi }(E_{Q}[Y],Q)
\end{equation*}
\end{proposition}

\begin{proof}
Let $Y\in \mathcal{K}$ and $\left\{ Y_{n}\right\} \subseteq \mathcal{K}$. If
$Y_{n}\geq Y$ and $Q\in \mathcal{M}$ then $E_{Q}[Y_{n}]\geq E_{Q}[Y].$ Since
$H(\cdot ,Q)$ is increasing, for every $Q\in \mathcal{M}$ we obtain
\begin{equation*}
H^{+}(E_{Q}[Y],Q):=\inf_{s>E_{Q}[Y]}H(s,Q)\leq \lim_{Y_{n}\downarrow
Y}H(E_{Q}[Y_{n}],Q).
\end{equation*}%
Therefore:%
\begin{align*}
\Psi _{\varphi }(Y)& =\inf_{Q\in \mathcal{M}}H(E_{Q}[Y],Q)\leq \inf_{Q\in
\mathcal{M}}H^{+}(E_{Q}[Y],Q)\leq \inf_{Q\in \mathcal{M}}\lim_{Y_{n}%
\downarrow Y}H(E_{Q}[Y_{n}],Q) \\
& =\lim_{Y_{n}\downarrow Y}\inf_{Q\in \mathcal{M}}H(E_{Q}[Y_{n}],Q)=%
\lim_{Y_{n}\downarrow Y}\Psi _{\varphi }(Y_{n})=\Psi _{\varphi }(Y).
\end{align*}%
The last equality in the Proposition follows from (\ref{HRinf}).
\end{proof}

\bigskip

In the following theorem we provide the representation of $\varphi $ in
terms of the dual functions $H_{\varphi }$ and $R_{\varphi }$ defined in (%
\ref{HH}) and (\ref{H2}).

\begin{theorem}
\label{thA}Suppose that $\mathcal{K\subseteq L}^{\infty }(\Omega ,\mathcal{F}%
)$ is a convex cone $\sigma (\mathcal{L}^{\infty },ca)$-closed and that $%
\varphi :\mathcal{K}\rightarrow \overline{\mathbb{R}}$ is monotone
increasing, quasiconcave and $\sigma (\mathcal{L}^{\infty },ca)$-upper
semicontinuous (using the relative topology on $\mathcal{K}$). Then for all $%
Y\in \mathcal{K}$
\begin{equation}
\varphi (Y)=\Psi _{\varphi }(Y)=\inf_{Q\in \mathcal{K}_{1}^{\circ
}}H_{\varphi }(E_{Q}[Y],Q)=\inf_{Q\in \mathcal{K}_{1}^{\circ }}R_{\varphi
}(E_{Q}[Y],Q).
\end{equation}
\end{theorem}

\begin{proof}
Fix $Y\in \mathcal{K}$. As $Y\in \left\{ \xi \in \mathcal{K}\mid E_{\mu
}[\xi ]\leq E_{\mu }[Y]\right\} $, by the definition of $H(E_{\mu }[Y],\mu )$
we deduce that, for all $\mu \in ca$
\begin{equation*}
H(E_{\mu }[Y],\mu )\geq \varphi (Y)
\end{equation*}%
hence
\begin{equation}
\inf_{\mu \in ca}H(E_{\mu }[Y],\mu )\geq \varphi (Y).  \label{AAA}
\end{equation}%
We prove the opposite inequality. Let $\varepsilon >0$ and define the set%
\begin{equation*}
C_{\varepsilon }:=\left\{ \xi \in \mathcal{K}\mid \varphi (\xi )\geq \varphi
(Y)+\varepsilon \right\}
\end{equation*}%
As $\varphi $ is quasi-concave and $\sigma (\mathcal{L}^{\infty },ca)$-upper
semicontinuous (on $\mathcal{K}$), $C_{\varepsilon }$ is convex and $\sigma (%
\mathcal{L}^{\infty },ca)$-closed. Suppose $\varphi (Y)>-\infty $ (if $%
\varphi (Y)=-\infty ,$ we may take $C_{M}:=\left\{ \xi \in L^{\infty }\mid
\varphi (\xi )\geq -M\right\} $ and the following argument would hold as
well). Since $Y\notin C_{\varepsilon }$, the Hahn Banach theorem implies the
existence of a continuous linear functional that strongly separates $Y$ and $%
C_{\varepsilon },$ that is there exist $\mu _{\varepsilon }\in ca$ such that%
\begin{equation}
E_{\mu _{\varepsilon }}[\xi ]>E_{\mu _{\varepsilon }}[Y]\text{ for all }\xi
\in C_{\varepsilon }.  \label{BBB}
\end{equation}%
Hence%
\begin{equation*}
\left\{ \xi \in \mathcal{K}\mid E_{\mu _{\varepsilon }}[\xi ]\leq E_{\mu
_{\varepsilon }}[Y]\right\} \subseteq C_{\varepsilon }^{c}:=\left\{ \xi \in
\mathcal{K}\mid \varphi (\xi )<\varphi (Y)+\varepsilon \right\}
\end{equation*}%
and from (\ref{AAA})%
\begin{align*}
\varphi (Y)& \leq \inf_{\mu \in ca}H(E_{\mu }[Y],\mu )\leq H(E_{\mu
_{\varepsilon }}[Y],\mu _{\varepsilon }) \\
& =\sup \left\{ \varphi (\xi )\mid \xi \in \mathcal{K}\text{ and }E_{\mu
_{\varepsilon }}[\xi ]\leq E_{\mu _{\varepsilon }}[Y]\right\} \\
& \leq \sup \left\{ \varphi (\xi )\mid \xi \in \mathcal{K}\text{ and }%
\varphi (\xi )<\varphi (Y)+\varepsilon \right\} \leq \varphi (Y)+\varepsilon
.
\end{align*}%
Therefore, $\varphi (Y)=\inf_{\mu \in ca}H(E_{\mu }[Y],\mu )$.

\bigskip

To show that the $inf$ can be taken over $\mathcal{K}^{\circ }$, it is
sufficient to prove that $\mu _{\varepsilon }\in \mathcal{K}^{\circ }$. Let $%
\xi \in C_{\varepsilon }.$ Given that $\varphi $ is monotone increasing and
that $\mathcal{K}$ is a convex cone, $\xi +nW\in C_{\varepsilon }$ for every
$n\in \mathbb{N}$ and $W\in \mathcal{K}^{+}$. From (\ref{BBB}), we have:%
\begin{equation*}
E_{\mu _{\varepsilon }}[(\xi +nW)]>E_{\mu _{\varepsilon }}[Y]\Rightarrow
E_{\mu _{\varepsilon }}[W]>\frac{E_{\mu _{\varepsilon }}[(Y-\xi )]}{n}%
\rightarrow 0,\text{ as }n\rightarrow \infty .
\end{equation*}%
As this holds for any $W\in \mathcal{K}^{+}$ we deduce that $\mu
_{\varepsilon }\in \mathcal{K}^{\circ }$. Therefore, $\inf_{\mu \in \mathcal{%
K}^{\circ }}H(E_{\mu }[Y],\mu )=\varphi (Y)$. By definition of $H(p,\mu )$,%
\begin{equation*}
H(E_{\mu }[Y],\mu )=H(E_{\lambda \mu }[Y],\lambda \mu )\quad \forall \mu \in
\mathcal{K}^{\circ }\text{ , }\mu \neq 0,\text{ }\lambda \in (0,\infty ).
\end{equation*}%
Hence we deduce $\varphi (Y)=\inf_{Q\in \mathcal{K}_{1}^{\circ
}}H(E_{Q}[Y],Q).$

The remaining equalities follows from Propositions (\ref{propHH}) and (\ref%
{propHR}), since $\varphi (Y)=\inf_{Q\in \mathcal{K}_{1}^{\circ
}}H(E_{Q}[Y],Q)$ is continuous from above. Indeed, if $Y_{n}\downarrow Y$
then the monotonicity of $\varphi $ implies $\varphi (Y_{n})\geq \varphi (Y)$
and $\varphi (Y_{n})\downarrow $. Moreover the Monotone Convergence Theorem
implies for every $\mu =\mu ^{+}-\mu ^{-}\in ca$ that $E_{\mu ^{\pm
}}[Y_{n}]\rightarrow E_{\mu ^{\pm }}[Y]$ and therefore $Y_{n}\rightarrow Y$
in the $\sigma (\mathcal{L}^{\infty },ca)$. As $\mathcal{K}$ is closed and $%
\varphi $ is usc in the $\sigma (\mathcal{L}^{\infty },ca)$ topology, we can
conclude $\limsup_{n}\varphi (Y_{n})=\lim \varphi (Y_{n})\leq \varphi (Y)$,
and therefore $\lim \varphi (Y_{n})=\varphi (Y)$.
\end{proof}

\section{Appendix B: Dual representation\label{AppB}}

We recall, from \cite{FMP12}, the dual representation of risk measures
defined on $\mathcal{P}(\mathbb{R}).$ Consider the set%
\begin{eqnarray*}
C_{b}^{-}(\mathbb{R}) &=&\left\{ f\in C_{b}(\mathbb{R})\mid f\text{ is
decreasing}\right\} . \\
&=&\left\{ f\in C_{b}(\mathbb{R})\mid Q,P\in \mathcal{P}(\mathbb{R})\text{
and }P\preccurlyeq _{1}Q\Rightarrow \int fdQ\leq \int fdP\right\} .
\end{eqnarray*}

\begin{proposition}[Prop. 5.6 \protect\cite{FMP12}]
\label{propvolle}Any $\sigma (\mathcal{P}(\mathbb{R}),C_{b}(\mathbb{R}))$%
-lsc Risk Measure $\Phi :\mathcal{P}(\mathbb{R})\rightarrow \mathbb{R}\cup
\left\{ \infty \right\} $ can be represented as
\begin{equation*}
\Phi (P)=\sup_{f\in C_{b}^{-}}V^{-1}\left( \int fdP,f\right) .
\end{equation*}%
where $V:\mathbb{R}\times C_{b}(\mathbb{R})\rightarrow \overline{\mathbb{R}}$
is given by:
\begin{equation}
V(a,f):=\sup_{Q\in \mathcal{P}(\mathbb{R})}\left\{ \int fdQ\mid \Phi (Q)\leq
a\right\} \text{, }a\in \mathbb{R},  \label{123}
\end{equation}%
and
\begin{equation}
V^{-1}(v,f):=\inf \left\{ \alpha \in \mathbb{R}\mid V(\alpha ,f)\geq
v\right\} ,\text{ }v\in \mathbb{R}.  \label{1111}
\end{equation}
\end{proposition}

We also mention that the $\sigma (\mathcal{P}(\mathbb{R}),C_{b}(\mathbb{R}))$
lower semicontinuity property can be characterized with an appropriate and
simple continuity from above condition with respect to the first order
stochastic dominance (see \cite{FMP12} Proposition 2.5).

\begin{example}[The certainty equivalent]
Fix any continuous, bounded from below and strictly decreasing function $f:%
\mathbb{R}\rightarrow \mathbb{R}$. Then the map $\Phi _{f}:\mathcal{P}(%
\mathbb{R})\rightarrow \mathbb{R}\cup \{+\infty \}$ defined by:
\begin{equation*}
\Phi _{f}(P):=-f^{-1}\left( \int fdP\right)
\end{equation*}%
is a Risk Measure on $\mathcal{P}(\mathbb{R})$. It is also easy to check
that $\Phi _{f}$ is $\sigma (\mathcal{P}(\mathbb{R}),C_{b}(\mathbb{R}))-$%
lsc. In \cite{FMP12} it is shown that $\Phi _{f}$ can not be convex on $%
\mathcal{P}(\mathbb{R})$. By selecting the function $f(x)=e^{-x}$ we obtain $%
\Phi _{f}(P)=\ln \left( \int \exp \left( -x)dF_{P}(x)\right) \right) $,
which is in addition (TrI). Its associated risk measure $\rho _{\mathbb{P}%
}:L^{0}(\Omega ,\mathcal{F},\mathbb{P})\rightarrow \mathbb{R}\cup \{+\infty
\}$ defined on random variables, $\rho _{\mathbb{P}}(X)=\Phi _{f}(P_{X})=\ln
E_{\mathbb{P}}[e^{-X}],$\ is the Entropic Convex Risk Measure.
\end{example}

\paragraph{On the $\Lambda V@R$\label{lambda}.}

All the details of the present section can be found in \cite{FMP12}. We
consider a family of risk measures called $\Lambda V@R$ which depend on a
Probability/Loss function $\Lambda $. This family provides one example of a $%
V\&R$ measure that exhibits the peculiar cash invariance property (\ref{Cash}%
).

Fix the \emph{right continuous} increasing function $\Lambda :\mathbb{R}%
\rightarrow \lbrack 0,1]$ and define the family $\left\{ F_{m}\right\}
_{m\in \mathbb{R}}$ of functions $F_{m}:\mathbb{R}\rightarrow \lbrack 0,1]$
by
\begin{equation*}
F_{m}(x):=\Lambda (x)\mathbf{1}_{(-\infty ,m)}(x)+\mathbf{1}_{[m,+\infty
)}(x).
\end{equation*}%
It is easy to show that if $\sup_{x\in \mathbb{R}}\Lambda (x)<1$ then the
associated map $\Phi :\mathcal{P}(\mathbb{R})\rightarrow \mathbb{R}\cup
\{+\infty \}$ defined by
\begin{equation*}
\Phi (P):=-\sup \left\{ m\in \mathbb{R}\mid P\in \mathcal{A}^{m}\right\} ,
\end{equation*}%
with
\begin{equation*}
\mathcal{A}^{m}:=\{Q\in \mathcal{P}(\mathbb{R})\mid F_{Q}\leq F_{m}\},
\end{equation*}%
is (Mon), (Qco) and $\sigma (\mathcal{P}(\mathbb{R}),C_{b})-$l.s.c.. This
map was named $\Lambda V@R$ since
\begin{equation*}
\Lambda V@R(P_{X}):=-\sup \left\{ m\in \mathbb{R}\mid \mathbb{P}(X\leq
x)\leq \Lambda (x),\;\forall x\leq m\right\} .
\end{equation*}%
If $\Lambda (x)=\lambda $ for every $x\in \mathbb{R}$ then%
\begin{equation*}
\Lambda V@R=V@R_{\lambda }
\end{equation*}%
coincides with the classical Value at Risk $V@R_{\lambda }$, and in
particular if $\lambda =0$ we recover the worst case risk measure
\begin{equation*}
\rho _{w}(X)=-\inf (X).
\end{equation*}%
Both these risk measures are (TrI), whereas for a general $\Lambda $ we get
the following property
\begin{equation}
\Lambda V@R(P_{X+\alpha })=\Lambda ^{\alpha }V@R(P_{X})-\alpha ,
\label{Cash}
\end{equation}%
where $\Lambda ^{\alpha }(x)=\Lambda (x+\alpha )$. Clearly by the definition
$\Lambda V@R(\delta _{0})=0$ for every $\Lambda $ so that
\begin{equation}
\Lambda V@R(\delta _{\alpha })=-\alpha  \label{pha}
\end{equation}%
We also mention that the $\Lambda V@R$ is elicitable (depending on the
selection of $\Lambda $) and is statistically consistent. We refer to \cite%
{BPR16} for details on this topic.

\noindent Regarding the dual representation of the $\Lambda V@R,$ the
functions $V$ in (\ref{123}) and $V^{-1}$ in (\ref{1111}) can be easily
computed (see \cite{FMP12}):%
\begin{eqnarray}
V(a,f) &=&f(-\infty )+\int_{-\infty }^{-a}(1-\Lambda )df,  \label{111} \\
V^{-1}(v,f) &=&-H_{f}^{l}(v-f(-\infty )),  \label{222}
\end{eqnarray}%
where $H_{f}^{l}$ is the left inverse of the function: $a\rightarrow
\int_{-\infty }^{a}(1-\Lambda )df$.

\noindent As two particular cases, from (\ref{111}) and (\ref{222}), we get
for the $V@R_{\lambda }$ (where $\Lambda (x)=\lambda )$ : $V^{-1}\left(
v,f\right) =-f^{l}\left( \frac{v-\lambda f(-\infty )}{1-\lambda }\right) $;
for the Worst Case risk measure $\rho _{w}$ (where $\Lambda (x)=0)$ we
obtain $V^{-1}\left( v,f\right) =-f^{l}(v)$, where $f^{l}$ is the left
inverse of $f$.

\end{document}